\documentclass[a4paper]{article} %
\usepackage[hyphens]{url}  %
\usepackage{graphicx} %
\usepackage{natbib}  %
\usepackage{caption} %
\usepackage{algorithm}
\usepackage{algorithmic}
\usepackage{xcolor}
\usepackage{comment}
\usepackage{amsmath, amssymb, mathrsfs, mathtools}
\usepackage{amsthm}
\usepackage{cleveref}
\usepackage{tikz}
\usepackage{xspace}
\usepackage{thmtools}
\usepackage{thm-restate}
\usepackage{amsthm}
\usepackage{listings}
\usepackage{nicefrac}
\usepackage{bbm}

\usepackage{enumitem}

\definecolor{codegreen}{rgb}{0,0.6,0}
\definecolor{codegray}{rgb}{0.5,0.5,0.5}
\definecolor{codepurple}{rgb}{0.58,0,0.82}
\definecolor{backcolour}{rgb}{0.95,0.95,0.92}
\definecolor{darkred}{rgb}{0.55, 0.0, 0.0}

\newtheorem{theorem}{Theorem}[section]
\newtheorem{lemma}[theorem]{Lemma}
\newtheorem{proposition}[theorem]{Proposition}
\newtheorem{corollary}[theorem]{Corollary}
\newtheorem{definition}{Definition}[section]
\newtheorem{example}{Example}[section]
\newtheorem{claim}{Claim}

\newenvironment{claimproof}{\par\noindent\emph{Proof. }}{\leavevmode\unskip\penalty9999 \hbox{}\nobreak\hfill\quad\hbox{$\blacksquare$}}

\DeclareMathOperator*{\argmax}{arg\,max}
\DeclareMathOperator*{\argmin}{arg\,min}

\newcommand{\PBrule}{R}
\newcommand{\rx}[1]{MES[#1]\xspace}
\newcommand{\ejr}[1]{#1-EJR\xspace}
\newcommand{\ejrsat}{\ejr{$\sat$}}
\newcommand{\ejro}[1]{#1-EJR-1\xspace}
\newcommand{\ejrosat}{\ejro{$\sat$}}
\newcommand{\ejrosatplus}{$\sat$-EJR-$1^+$}
\newcommand{\ejrx}[1]{#1-EJR-x\xspace}
\newcommand{\ejrxsat}{\ejrx{$\sat$}}
\newcommand{\pjr}[1]{#1-PJR\xspace}
\newcommand{\pjrsat}{\pjr{$\sat$}}
\newcommand{\pjro}[1]{#1-PJR-1\xspace}
\newcommand{\pjrosat}{\pjro{$\sat$}}
\newcommand{\pjrx}[1]{#1-PJR-x\xspace}
\newcommand{\pjrxsat}{\pjrx{$\sat$}}
\newcommand{\sat}{\mu}
\newcommand{\satNum}{\sat^{\#}}
\newcommand{\satCost}{\sat^{c}}
\newcommand{\satSqrt}{\sat^{\sqrt{c}}}

\usepackage[textsize=tiny]{todonotes}
\usepackage{newfloat}
\usepackage{listings}
\floatstyle{ruled}
\newfloat{listing}{tb}{lst}{}
\floatname{listing}{Listing}
\title{Proportionality in Approval-Based Participatory Budgeting}
\author{
    Markus Brill,\textsuperscript{\rm 1,2} Stefan Forster,\textsuperscript{\rm 3} Martin Lackner,\textsuperscript{\rm 3} Jan Maly,\textsuperscript{\rm 3} Jannik Peters\textsuperscript{\rm 1}
}
\date{
    \textsuperscript{\rm 1}TU Berlin, Berlin, Germany\\
    \textsuperscript{\rm 2}University of Warwick, Coventry, UK\\
    \textsuperscript{\rm 3}TU Wien, Vienna, Austria
}

\begin{document}

\maketitle

\begin{abstract}
The ability to measure the satisfaction of (groups of) voters is a crucial prerequisite for formulating proportionality axioms in approval-based participatory budgeting elections. Two common\,---\,but very different\,---\,ways to measure the satisfaction of a voter consider \textit{(i)} the number of approved projects and \textit{(ii)} the total cost of approved projects, respectively.  
In general, it is difficult to decide which measure of satisfaction best reflects the voters' true utilities.
In this paper, we study proportionality axioms with respect to large classes of \emph{approval-based satisfaction functions}. We establish logical implications among our axioms and related notions from the literature, and we ask whether outcomes can be achieved that are proportional with respect to more than one satisfaction function. We show that this is impossible for the two commonly used satisfaction functions when considering proportionality notions based on \textit{extended justified representation}, but achievable for a notion based on \textit{proportional justified representation}. For the latter result, we introduce a strengthening of priceability and show that it is satisfied by several polynomial-time computable rules, including the Method of Equal Shares and Phragmén's sequential rule. 
\end{abstract}

\newcommand{\dns}{DNS\xspace}

\section{Introduction}

``How can cities ensure that the results of their participatory budgeting process proportionally represents the preferences of the citizens?'' This is the key question in a recently emerging line of research on proportional participatory budgeting \citep{aziz:proportionality, PPS21a,los:proportional}. 
Participatory budgeting (PB) is the collective process of identifying a set of projects to be realized with a given budget cap; often, the final decision is reached by voting \citep[e.g.,][]{Laru21a}. 
The goal of \textit{proportional} PB is to identify voting rules that guarantee
proportional representation without the need to declare \textit{a priori} which groups deserve representation. Instead, each group of sufficient size with sufficiently similar interests is taken into account.
Such a group could be a district, cyclists, parents, or any other collection of people with similar preferences.
This is contrast to, e.g., assigning each district a proportional part of the budget, which excludes other (cross-district) groups from consideration.

To be able to speak about proportional representation in the context of PB, one first needs to decide on how to measure the representation of a given voter by a selection of projects. If votes are cast in the form of approval ballots, as is the case in most PB processes in practice, 
two standard ways to measure the satisfaction of a voter have emerged. The first assumes that the satisfaction a voter derives from an outcome is the \textit{total cost of the approved projects} in this outcome \citep{aziz:proportionality, AzLe21, TF19a}. 
In other words, voters care about how much money is spent on projects they like. 
The second assumes the satisfaction of a voter to be simply the \textit{number of approved projects} in the outcome \citep{PPS21a, los:proportional, FVMG22a, TF19a}. 
We refer to these two measures as \textit{cost-based satisfaction} and \textit{cardinality-based satisfaction}, respectively. 
Both measures, though naturally appealing, have their downsides: Under the cost-based satisfaction measure, inefficient (i.e., more expensive) projects are seen as preferable to equivalent but cheaper ones.
Under the cardinality-based satisfaction measure, large projects (e.g., a new park) and small projects (e.g., a new bike rack) are treated as equivalent.

The ambiguity of measuring satisfaction leads to three main problems:
First, different papers present incomparable notions of fairness
based on different measures of satisfaction. For example, both \citet{aziz:proportionality}
and \citet{los:proportional} generalized a well-known proportionality axiom known as proportional justified representation (PJR), but they did so based
on different satisfaction measures. 
Second, the two measures described above 
are certainly not the only reasonable functions for measuring satisfaction; 
and results in the literature cannot easily be transferred to new satisfaction functions. 
For example, satisfaction could be estimated by experts evaluating projects; if efficiency is taken into account, such a measure may differ significantly from the cost-based one.
Third, most papers so far have focused on a single
satisfaction function only. Therefore, it is not known whether we can guarantee proportionality properties with respect to different satisfaction measures simultaneously. This would 
be extremely useful in practice: If a mechanism designer is not sure which satisfaction function most accurately describes the voters' preferences in a given PB process, she could potentially choose
a voting rule that provides proportionality guarantees with respect to all satisfaction functions that seem plausible to her.

\subsubsection{Our contribution.}

To tackle these problems, we propose a general framework for studying proportionality in approval-based participatory budgeting:
We employ the notion of \textit{(approval-based) satisfaction functions} \citep{TF19a}, i.e., functions that, for every possible outcome, assign to each voter a satisfaction value based on the voter's approval ballot. We then use this notion of satisfaction functions to unify the different proportionality notions studied by  \citet{aziz:proportionality}, \citet{PPS21a}, and \citet{los:proportional} into one framework and analyze their relations.%

Furthermore, we identify a large class of satisfaction functions that are of particular interest: \textit{Weakly decreasing normalized satisfaction} (short: DNS) functions are satisfaction functions for which more expensive projects offer at least as much satisfaction as cheaper projects, but the satisfaction does not grow faster than the cost. Intuitively, the cardinal measure is one extreme of this class (the satisfaction does not change with the cost) while the cost-based measure is the other extreme (the satisfaction grows exactly like the cost). For each satisfaction function in this class, we show that an instantiation of the \textit{Method of Equal Shares} (MES) \citep{PeSk20a, PPS21a} satisfies \textit{extended justified representation up to \textit{any} project (EJR-x)}.\footnote{This strengthens a result by \citet{PPS21a}, showing that MES satisfies EJR up to \textit{one} project (EJR-1) for additive utility functions.} 
However, while MES for a \textit{specific} satisfaction function satisfies EJR-x, we can show that even the weaker notion of EJR-1 is incompatible for the cost-based and cardinality-based satisfaction functions. 
In other words, it is not possible to find a voting rule that guarantees EJR-1 for the cost-based and the cardinality-based satisfaction measure simultaneously.

To deal with this incompatibility, we turn to the notion of \textit{proportional justified representation (PJR)} and show that a specific class of rules, including sequential Phragm\'{e}n and one variant of MES, satisfies \textit{PJR up to any project (PJR-x)} for \textit{all} \dns satisfaction functions at once. In other words, when using one of these rules, we generate an outcome that can be seen as proportional no matter which satisfaction function is used, as long as the function is a \dns satisfaction function.

\subsubsection{Related work.}

The study of proportional PB crucially builds on the literature on approval-based committee voting \citep{lackner:abc}. The proportionality notions most relevant to our paper are \textit{extended justified representation (EJR)} \citep{aziz:justified}, \textit{proportional justified representation (PJR)} \citep{fernandez:proportionalJustified}, and \textit{priceability} \citep{PeSk20a}.

Proportionality in PB was first considered by \citet{aziz:proportionality}, who generalized PJR as well as the maximin support method \citep{SFFB18a}.
This setting was subsequently generalized to voters with ordinal preferences \citep{AzLe21}.
The concept of satisfaction functions was introduced by \citet{TF19a}, 
who presented a framework for designing (non-proportional) approval-based PB rules. Besides the cost-based and the cardinality-based satisfaction function, they also studied a satisfaction measure based on the Chamberlin--Courant method \citep{ccElection}. 

\citet{PPS21a} studied PB with arbitrary additive utilities and showed that a generalized variant of the \textit{Method of Equal Shares (MES)} \citep{PeSk20a} satisfies EJR \textit{up to one project}. 
The approval-based satisfaction functions studied in our paper constitute special cases of additive utility functions, and the additional structure provided by this restriction allows us to show a significantly stronger result.

\citet{los:proportional} study the logical relationship of proportionality axioms in PB with either additive utilities or the cardinality-based satisfaction function. They generalize notions such as PJR, laminar proportionality, and priceability to the two aforementioned settings and study how MES, sequential Phragm\'{e}n, and other rules behave with regard to these axioms. In particular, they show that sequential Phragm\'{e}n satisfies PJR for the cardinality-based satisfaction function. We strengthen the latter result along multiple dimensions, by identifying a class of rules satisfying PJR-x for a whole class of satisfaction functions simultaneously. ({PJR-x is equivalent to PJR for the cardinality-based satisfaction function.})

Besides proportionality, other recent topics in PB include the handling of donations \citep{CLM21a}, the study of districts \citep{HKPP21a} and projects groups \citep{JSTZ21a}, the maximin objective \citep{SBN22a}, welfare/representation trade-offs \citep{FVMG22a}, and uncertainty in the cost of projects \citep{BBL22a}.

\section{Preliminaries}

For $t \in \mathbb{N}$, we let $[t]$ denote the set $[t] = \{1, \dots, t\}$.
        
Let $N = [n]$ be a set of $n$ \emph{voters} and $P = \{p_1, \dots, p_m\}$ a set of $m$ projects. Each voter $i \in N$ is associated with an \emph{approval ballot} $A_i \subseteq P$ and an \emph{approval profile} $A = (A_1, \dots, A_n)$ lists the approval ballots of all voters. Further, $c\colon P \to \mathbb{R}^+$ is a \textit{cost function} mapping each project $p \in P$ to its \emph{cost} $c(p)$. Finally, $b \in \mathbb{R}^+$ is the \emph{budget limit}.
        
Together, $(A, P, c, b)$ form an \emph{approval-based budgeting (ABB)} instance. 
For a subset $W \subseteq P$ of projects, we define $c(W) = \sum_{p \in W} c(p)$. 
We call $W$ an \emph{outcome} if \mbox{$c(W) \le b$}, i.e., if the projects in $W$ together cost no more than the budget limit. Further, we call an outcome $W$ \emph{exhaustive} if there is no outcome $W' \supset W$. 
An ABB rule $R$ now assigns every ABB instance $E = (A, P, c, b)$  to a non-empty set $R(E)$ of outcomes. If every outcome in $R(E)$ is exhaustive for every ABB instance $E$, we call the rule $R$ \emph{exhaustive}.
        
For a project $p \in P$ we let $N_p \coloneqq \{i \in N \colon p \in A_i\}$ denote the set of \emph{approvers} of $p$. We often write $N_j$ for~$N_{p_j}$.
        
An ABB instance with $c(p) = 1$ for all $p \in P$ is called a \emph{unit-cost} instance and corresponds to an approval-based committee voting instance with $\lfloor b \rfloor$ seats. %

Next, we define our key concept. 

\begin{definition}\label{def:sat-fct}
Given an ABB instance
$(A, P, c, b)$, an \emph{(approval-based) satisfaction function} is a function
$\sat\colon 2^P \rightarrow \mathbb{R}_{\ge 0}$ that satisfies the following conditions: $\sat(W) \leq \sat(W')$ whenever $W \subseteq W'$ and $\sat(W) = 0$ if and only if $W = \emptyset$.
\end{definition}

The satisfaction $\sat_i(W)$ that a voter $i$ derives from an outcome $W \subseteq P$ with respect to the 
satisfaction function $\sat$ is defined as the satisfaction generated
by the projects in $W$ that are approved by $i$, i.e.,
\[\sat_i(W) = \sat(A_i \cap W).\]
For notational convenience, we write $\sat(p)$ instead of
$\sat(\{p\})$ for an individual project $p \in P$.

Some of our results holds for restricted classes of satisfaction
functions. In particular, we are interested in the following properties.

\begin{definition}
Given an ABB instance $(A, P, c, b)$, a satisfaction function $\sat$ is
\begin{itemize}
\item \emph{additive} if $\sat(W) = \sum_{p_i\in W}\sat(p_i)$ for all $W \subseteq P$.
\item \emph{strictly increasing} if $\sat(W) < \sat(W')$ for all $W, W' \subseteq P$ with $W \subset W'$. 
\item \emph{cost-neutral} if $\sat(W) = \sat(W')$ for all $W, W' \subseteq P$ 
such that there is a bijection $f\colon W \to W'$ for which $c(p) = c(f(p))$ holds for all $p \in P$. 
\end{itemize}
\end{definition}

Clearly, every additive satisfaction function is also
strictly increasing. 
The two most prominent satisfaction functions are the following.

\begin{definition}\label{models:sat:definition:w-and-a}
Given an ABB instance $(A, P, c, b)$ and a set $W \subseteq P$,
           the \emph{cost-based satisfaction function} $\satCost$ is defined as 
\(\satCost(W) = c(W) = \sum_{p\in W}  c(p) \)
               and the \emph{cardinality-based satisfaction function} $\satNum$ is defined as 
\(\satNum(W) = |W|.  \)
        \end{definition}

\noindent Clearly, $\satCost$ and $\satNum$ are cost-neutral and additive.

An example for a cost-neutral 
satisfaction function that is not strictly increasing (and, hence, not additive) is the \textit{CC satisfaction function} \citep{TF19a}, which is inspired by the well-known Chamberlin--Courant rule \citep{ccElection}:
\[\sat^{CC}(W)=  \begin{cases}
0 &\text{if } W = \emptyset\\
1 &\text{otherwise.}
\end{cases}\]
An example for an additive satisfaction function that is not cost-neutral
is \textit{share} \citep{lackner:fairness}: %
\[\sat^{\text{share}}(W) = \sum_{p \in W} \frac{c(p)}{|N_p|}.\]

We illustrate the two most prominent satisfaction functions,
$\satCost$ and $\satNum$, with a simple example.

\begin{example}
Consider an ABB instance
with one voter, five projects, and budget $b = 5$; the voter approves all projects and the cost of each project is $1$ except the first project, which has cost $c(p_1)=5$. 
Under $\satCost$ the best outcome is $\{p_1\}$, which gives the voter a  satisfaction of $5$.
Under $\satNum$, the best outcome is $\{p_2, \dots, p_5\}$,
with a satisfaction of $4$.
\end{example}

Let us show a simple lemma about the unit-cost case that we are going to use repeatedly. 

\begin{lemma} \label{lem:1cost}
Consider a unit-cost ABB instance and a cost-neutral and strictly increasing satisfaction function $\sat$. Then, the following equivalence holds for all outcomes $W,W'$: 
\[\sat(W)\ge \sat(W') \quad \text{if and only if} \quad |W| \ge |W'|.\]
\end{lemma}

\begin{proof}
From right to left: Consider $W, W'$ such that $|W| \ge |W'|$. If $|W| = |W'|$, then by cost-neutrality and unit costs we have $\sat(W) = \sat(W')$. If $|W| > |W'|$, the previous argument and strict increasingness imply $\sat(W) > \sat(W')$.

From left to right: Consider $W, W'$ such that $\sat(W) \ge \sat(W')$. Assume for contradiction that $|W'|>|W|$. Then, by the argument above, we would have $\sat(W') > \sat(W)$, a contradiction. It follows that $|W| \ge |W'|$.  
\end{proof}

Next, we define a natural subclass of additive and cost-neutral satisfaction functions that
contains both $\satCost$ and $\satNum$. An additive satisfaction function
belongs to this class if
\textit{(i)} more expensive projects provide at least as much satisfaction as cheaper ones, and \textit{(ii)} more expensive projects do not provide a higher satisfaction per cost than cheaper projects.

\begin{definition}
Consider an ABB instance $( A, P, c, b )$. 
An additive satisfaction function $\sat$ has 
\emph{weakly decreasing normalized satisfaction (\dns)} if for all projects
$p,p' \in P$ with $c(p) \leq c(p')$ the following two inequalities hold:
\[
\sat(p) \leq \sat(p') \qquad \text{and} \qquad \text{$\frac{\sat(p)}{c(p)} \geq \frac{\sat(p')}{c(p')}$.}
\]
In this case, we call $\sat$ a \emph{\dns function}.
\end{definition}
Clearly, both $\satCost$ and $\satNum$ are \dns functions.
Indeed, they can be seen as two extremes among 
\dns functions since
$\satNum(p) = \satNum(p')$ holds for all $p, p'$, whereas for $\satCost$ we have
$\frac{\satCost(p)}{c(p)} = \frac{\satCost(p')}{c(p')}$.
Other natural examples of \dns functions 
include $\satSqrt(W) \coloneqq \sum_{p\in W} \sqrt{c(p)}$ and
$\sat^{\log(c)} \coloneqq \sum_{p\in W} \log(1+c(p))$. %

\smallskip

Finally, let us define an ABB rule that we use throughout the paper: the \emph{Method of Equal Shares (MES)}. In fact, we do not only define one rule, but rather a family of variants of MES, parameterized by a satisfaction function. We follow the definition of MES by \citet{PPS21a} in the setting of additive PB.

	\begin{definition}[\rx{$\sat$}]
	Given an ABB instance $( A, P, c, b)$ and a
satisfaction function $\sat$,  \rx{$\sat$} constructs an outcome $W$,
initially empty, iteratively as follows. It begins by assigning a budget of $b_i = \frac{b}{n}$ to each voter $i \in N$. A project $p_j \notin W$ is  called $\rho$-affordable if 
\[
\sum_{i \in N_j} \min(b_i, \rho \mu(p_j)) = c(p_j).
\]
In each round, the project $p_j$ which is $\rho$-affordable for the minimum $\rho$ is selected and for every $i \in N_j$, the budget $b_i$ is updated to 	$b_i - \min(b_i, \rho \mu(p_j))$.
This process is iterated until no further $\rho$-affordable projects are left (for any $\rho$).
\end{definition}

Intuitively, the parameter $\rho$ tells us how many units of budget a voter has to
pay for one unit of satisfaction. 

\section{Extended Justified Representation}\label{sec:ejr}

We begin our study of proportionality with the strong notion of extended justified representation (EJR). This concept was first introduced in the multiwinner setting by \citet{aziz:justified}.
On a very high level, it states that every group that is sufficiently ``cohesive'' deserves a
certain amount of representation in the final outcome. 
Therefore, we first need to define what it means for a group of voters in a PB instance
to be cohesive. For this, we follow \citet{PPS21a} and \citet{los:proportional}.\footnote{
\citet{aziz:proportionality} define cohesiveness slightly differently, which leads to
slightly different looking definitions of the axioms. The resulting definitions are, 
however, equivalent.}
 
    \begin{definition} \label{var:representation:definition:cohesiveness}
        Given an ABB instance $( A, P, c, b)$ and a set $T \subseteq P$ of projects, a subset $N' \subseteq N$ of voters is \emph{$T$-cohesive} if and only if $T \subseteq \bigcap_{i \in N'} A_i$ and $c(T) \leq \frac{|N'|}{n} b.$

    \end{definition}

Using this definition, we can now define EJR, which essentially states that in every $T$-cohesive 
group there is at least one voter that derives at least as much satisfaction from the
outcome as from $T$.

    \begin{definition}\label{var:control:EJR:definition:EJR}
       Given an ABB instance $( A, P, c, b )$ and a satisfaction function $\sat$, an outcome  $W \subseteq P$ satisfies \emph{extended justified representation with respect to~$\sat$} (\ejr{$\sat$}) if and only if 
       for any $T$-cohesive $N' \subseteq N$,
       there is some $i \in N'$ such that $\sat_i(W) \geq \sat_i(T)$. 

    \end{definition}
In the following we say that an ABB rule $\PBrule$ satisfies a property (in this case \ejrsat) if and only if, for every ABB instance $( A, P, c, b )$, each outcome in $R(A, P, c, b)$ satisfies this property. %
\Cref{var:control:EJR:definition:EJR} defines a whole class of axioms, one for each satisfaction function $\sat$.
This in contrast to the unit-cost setting, where only one version of the EJR axiom exists. 
This can be explained by the fact that \ejr{$\sat$} and \ejr{$\sat'$} are equivalent in the unit-cost setting for many satisfaction functions $\sat$ and $\sat'$.

    \begin{restatable}{proposition}{ejrunitequiv}
      Consider a unit-cost ABB instance and two additive and cost-neutral satisfaction functions~$\sat$ and~$\sat'$. Then, an outcome satisfies 
    \ejr{$\sat$} if and only if it satisfies \ejr{$\sat'$}.
    \end{restatable}

\begin{proof}
Assume that $\sat$ is a cost-neutral, strictly increasing
satisfaction function. To show our statement, we show that \ejrsat is equivalent to EJR in the unit cost setting \citep{aziz:justified}. In the follwing, we refer to this axiom as \textit{unit-EJR}.
An outcome $W$ satisfies unit-EJR if every $\ell$-cohesive\footnote{In the terminology of approval-based committee voting, a group is $\ell$-cohesive for a natural number $\ell$ if it is $T$-cohesive for some $T \subseteq P$ with $|T|=\ell$.} group contains a voter $i$ such that $|A_i \cap W| \geq \ell$.  

First, let $W \subseteq P$ be an outcome that satisfies unit-EJR. We show that $W$ satisfies \ejrsat as well.  
Let $N'$ be a $T$-cohesive group. For $\ell=|T|$, this group is $\ell$-cohesive
 (in unit-EJR terminology). Since $W$ satisfies unit-EJR, we know that that $|A_i \cap W| \geq \ell = |T|$ for some $i \in N'$. By \Cref{lem:1cost}, this implies $\sat(A_i \cap W) \geq \sat(T)$, or, equivalently, $\sat_i(W) \geq \sat_i(T)$.
Since $N'$ was chosen arbitrarily, this means that \ejrsat is satisfied.

Next, let $W \subseteq P$ be an outcome that does \textit{not} satisfy \mbox{unit-EJR}.
Then, there is a $\ell$-cohesive group $N'$ for which $|A_i \cap W| < \ell$ for all $i \in N'$. By definition, there exists a \mbox{$T \subseteq P$} such that $S$ is $T$-cohesive and $|T| = \ell$.
By \Cref{lem:1cost}, we have $\sat(A_i \cap W) < \sat(T)$ for all $i \in N'$, or equivalently, $\sat_i(W) < \sat_i(T)$ for all  $i \in N'$.
Therefore, \ejrsat is violated as well.
\end{proof}

Moreover, under these assumptions, \ejr{$\sat$} is equivalent to EJR as originally defined originally by \citet{aziz:justified}.
    By contrast, this is not the case, e.g., for \ejr{$\sat^{\text{CC}}$}.

Next, we show that \ejrsat is always satisfiable. Our proof adapts a similar
proof for general additive utility functions \citep{PPS21a} and employs the so-called Greedy Cohesive Rule.\footnote{Our result is
less general in that it only considers the approval case and more general in that it does
not assume additivity.}

    \begin{definition}
       \label{var:EJR:definition:greedycohesive}
       Given an ABB instance $( A, P, c, b )$, 
       the \emph{Greedy Cohesive Rule with respect to $\sat$ (\textsc{GCR[$\sat$]})} selects all outcomes $W$ that may result from the following procedure:
       
       \smallskip
       
       \begin{lstlisting}[escapeinside={(*}{*)}, language=Python]    
(*$W$*) := (*$\emptyset$*)
(*$N^{W}$*) := (*$N$*)
(*$\mathcal{C}(W, N^W)$*) := (*$\{W' \mid W' \subseteq P\backslash W \text{ and there exists } \\
\text{a $W'$-cohesive group } N' \subseteq N^W\}$*)
while (*$\mathcal{C}(W, N^W) \neq \emptyset$*)
    (*$W'$*) := (*$\argmax_{\widehat{W} \in \mathcal{C}(W,N^W)}\sat(\widehat{W})$*)
    (*$W$*) := (*$W \cup W'$*)
    (*$N^W$*) := (*$N^W \backslash N'$*)
    (*$\mathcal{C}(W, N^W)$*) := (*$\{W' \mid W' \subseteq P \backslash W \mbox{ and there exists } \\ 
    \text{a $W'$-cohesive group } N' \subseteq N^W\}$*)
return (*$W$*)
\end{lstlisting}
    \end{definition}

    \noindent This version of \textsc{GCR[$\sat$]} differs from the corresponding rule introduced by~\citet{PPS21a}, as they assume a tie-breaking mechanism for equally good subsets of $P \backslash W$ in favor of subsets of lower costs while in our rule the tie-breaking is done in an arbitrary fashion. This was purely done for ease of exposition, as we believe that breaking ties in favor of lower costs would have resulted in a algorithm of reduced readability. Another reason for our modified version of \textsc{GCR} is that it is already strong enough to always satisfy \ejrsat for any -- even non-additive -- satisfaction function $\sat$. The proof is essentially the same as the proof due to~\citet{PPS21a} because neither tie-breaking in favor of smaller costs nor additivity of satisfaction is assumed anywhere in the proof.

    \begin{restatable}{theorem}{greedyc}\label{var:control:EJR:EJR-s-and-BPJR-s-satisfiable}
    \ejrsat is always satisfiable for any satisfaction function $\sat$.
    \end{restatable}

\begin{proof}
We show that the output of GCR[$\sat$] always is an outcome that satisfies \ejrsat.
    
    \smallskip
    \textit{Feasibility:} Assume $W$ is a set of projects computed by \textsc{GCR[$\sat$]}. For any $W'$ that we chose at line 5, we know by the definition of cohesiveness that
    
    \[ c(W') \leq \frac{|S|\cdot b}{n} \]
    
    and since any voter in $N$ is only removed as a member of one such group, we know that for the cost of the outcome $c(W)$ it holds that
    
    \[ c(W) \leq \frac{|N|\cdot b}{n} = b \]
    
    and thus $W$ is a feasible outcome.
   
    \smallskip
    \textit{Satisfaction:} Assume to the contrary that there is a satisfaction function $\sat$ and ABB instance $( A, P, c, b )$, such that some outcome $W$ in \textsc{GCR[$\sat$]}($A, P, c, b$) does not satisfy \ejrsat. This means that by the definition of \ejrsat there is a $T$-cohesive group $N' \subseteq N$, such that for any $i \in N'$ we have
    
    \begin{align}\label{proof:gcr:1}
        \sat_i(W) < \sat(T)
    \end{align}
    
    By the strict increasingness of satisfaction functions this implies $T \not\subseteq W$.
    
    Now suppose $i_0$ is the first voter in $N'$ removed by the algorithm. Such a voter must exist, since otherwise we know that it holds that $N' \subseteq N^W$ at the last execution of the loop condition and thus $T$ could be added to $W$, which in turn means it was not the last execution of line 3.
    
    So $i_0$ was removed as a member of some $T'$-cohesive group $N''$. Since $T'\subseteq W$ and by $T'$-cohesiveness, we know that for any $i \in N''$
    \[\sat_i(W) \geq \sat(T')\]
    and since $i_0 \in N''$ thus
    $$\sat_{i_0}(W) \geq \sat(T') $$
    
    Since the algorithm chose $T'$ and $N''$ instead of $T$ and $N'$ at line 4, we know by line 6 that
    \[ \sat(T') \geq \sat(T) \]
    and thus for $i_0$ it holds that
    \[\sat_{i_0}(W) \geq \sat(T)\]
    contradicting (\ref{proof:gcr:1}).
    \end{proof}

The Greedy Cohesive Rule that is used to prove  \Cref{var:control:EJR:EJR-s-and-BPJR-s-satisfiable} has exponential running time. This is however unavoidable, as we can show that no algorithm can find an allocation satisfying \ejrsat in polynomial time (unless P = NP), for a large class of approval-based satisfaction functions. We call this class strictly cost-responsive.

\begin{definition}
We say that a satisfaction function $\sat$ is \emph{strictly cost-responsive}
if for all $W, W' \subseteq P$ with $c(W) < c(W')$, we have $\sat(W) < \sat(W')$.
\end{definition}

This class includes $\satCost$ but also functions with
diminishing (but not vanishing) marginal satisfaction like~$\satSqrt$.

\begin{restatable}{theorem}{nphardejr}\label{thm:NP-EJR}
Let $\sat$ be a satisfaction function
that is strictly cost-responsive for instances with a single voter.
Then, there is no polynomial-time algorithm that, given an ABB 
instance $( A, P, c, b )$ as input, always computes an outcome
satisfying \ejrsat, unless $P = NP$.
\end{restatable}

\begin{proof}
		
Assume that there is an algorithm $\mathbb{A}$ that always computes an
allocation satisfying \ejrsat.

We will make use of the \textsc{Subset-Sum} problem, which is known to be NP-hard.
In this problem, we are given as input a set $S = \{s_1, \ldots, s_m\}$ of integers and 
a target $t \in \mathbb{N}$ and we wonder whether there exists an $X \subseteq S$ 
such that $\sum_{x \in X} x = t$.
		
Given $S$ and $t$ as described above, we construct an ABB instances as follows.
We have $m$ projects $P = \{p_1, \ldots, p_m\}$ with the following cost function $c(p_j) = s_j$
for all $j \in \{1, \ldots, m\}$ and a budget limit $b = t$.
There is moreover only one voter, who approves of all the projects.
		
Now, $(S,t)$ is a positive instance of \textsc{Subset-Sum} if and only if there
is an outcome $T$ that cost is exactly $b$.
If such an allocation $T$ exists, then the one voter $1$ is $T$-cohesive.
Therefore, any allocation $W$ that satisfies \ejrsat must give that voter
$\sat(W) \geq \sat(T)$. By strict cost-responsiveness, this implies that
$c(W) \geq c(T) = b$.
Hence, $(S,t)$ is a positive instance of \textsc{Subset-Sum} if and only if
$c(\mathbb A(( A, P, c, b ))) = b$. This way, we can
use $\mathbb{A}$ to solve \textsc{Subset-Sum} in polynomial time.
\end{proof}

Note that $\satNum$ does not satisfy
strict cost-responsiveness. Indeed, outcomes satisfying \ejr{$\satNum$} can be computed efficiently, e.g., by employing \rx{$\satNum$} \citep{PPS21a, los:proportional}.
Further, we note that our reduction does not preclude efficient algorithms in the case that costs are bounded. Hence, it is open whether a pseudopolynomial-time algorithm exists.

Theorem~\ref{thm:NP-EJR} motivates us to consider weakenings of EJR.
First, we define EJR up to one project \citep{PPS21a}.

\begin{definition}\label{var:representation:definition:EJR-1}
 Given an ABB instance $( A, P, c, b )$ and a satisfaction function $\sat$, an outcome  $W \subseteq P$ satisfies \emph{EJR up to one project with respect to $\sat$} (\ejrosat) if and only if, 
 for every $T$-cohesive group $N'$, either $T \subseteq W$
or there exists a voter $i \in N'$ and a project $p \in P \setminus W$ such that 
	\(\sat_i(W \cup \{p\}) > \sat_i(T).\)

\end{definition}

\citet{PPS21a} have shown that we can satisfy \ejrosat for every additive satisfaction function $\sat$ using \rx{$\sat$}.\footnote{In the approval-based setting considered in this paper,
this is even true if we strengthen \ejrosat by requiring that the project $p$ comes from $T$,
i.e., by replacing $p \in P \setminus W$ with $p \in T \setminus W$ in Definition~\ref{var:representation:definition:EJR-1} (see 
the full version of this paper 
for details).}
Since the approval-based setting studied in this paper is a special case of the setting studied by \citet{PPS21a}, we can improve upon their result.  
Similar to the fair division literature, where the notion of envy-freeness up to one good (EF-1) can be strengthened to envy-freeness up to \textit{any} good (\mbox{EF-x}) \citep{CKMPSW19},
we strengthen \ejrosat to~\mbox{\ejrxsat}: Instead of requiring that there exists one project whose addition lets voter $i$'s satisfaction exceed $\sat(T)$, we require that this holds for \textit{every} unchosen project from $T$. %

\begin{definition}\label{var:representation:definition:EJR-x}
Given an ABB instance $( A, P, c, b )$ and a satisfaction function $\sat$, an outcome  $W \subseteq P$ satisfies \emph{EJR up to any project with respect to $\sat$} (\ejrxsat)
if and only if, 
for every $T$-cohesive group $N'$,  %
there is a voter $i \in N'$ such
that
	\(\sat_i(W \cup \{p\}) > \sat_i(T)\)
for every project $p \in T \setminus W$.

\end{definition}
By definition, \ejrxsat implies \ejrosat and, 
intuitively, we would assume that \ejrxsat is implied by \ejrsat.
This is indeed the case, at least
for strictly increasing satisfaction functions.
Moreover, \ejrsat, \ejrosat and \ejrxsat are equivalent in the unit-cost setting
as long as $\sat$ is strictly increasing and cost-neutral.
 
\begin{restatable}{proposition}{propejr} \label{var:representation:proposition:EJR-implies-EJR-x}
Let $\sat$ be a strictly increasing satisfaction function. Then,
\begin{enumerate}
    \item[(i)] \ejrsat implies \ejrxsat, and 
    \item[(ii)] for unit-cost instances if $\sat$ is cost-neutral, both \ejrosat and \ejrxsat are equivalent to \ejrsat.
\end{enumerate}

    \end{restatable}
 
\begin{proof}

For \textit{(i)}, assume that $W$ satisfies \ejrsat. This means that for any $T$-cohesive $N'$, there is some $i \in N'$ with
           \[ \sat_i(A_i \cap W) \geq \sat(T). \]
By strict increasingness, for any $p \in \bigcap_{i \in S}A_i \setminus W$ we have
            \[ \sat(\bigcup_{i \in S} A_i \cap W \cup \{p\}) > \sat(\bigcup_{i \in S} A_i \cap W) \geq \sat(T),\]
and thus \pjrxsat holds.

\smallskip

For \textit{(ii)}, consider a unit-cost instance $( A, P, c, b)$, an outcome $W \subseteq P$, and a $T$-cohesive group $N'$. If $T\subseteq W$, the requirements of \ejrsat, \ejrosat, and \ejrxsat are all satisfied.\footnote{Note that there is a typo in \Cref{var:representation:definition:EJR-1} in the main text. The correct definition is as follows: An outcome  $W \subseteq P$ satisfies \ejrxsat if and only if, for every $T$-cohesive group $N'$, either $T \subseteq W$ or there is a voter $i \in N'$ such that $\sat_i(W \cup \{p\}) > \sat_i(T)$ for every project $p \in T \setminus W$. (The part ``either $T \subseteq W$ or'' is missing in \Cref{var:representation:definition:EJR-1}.)} 
Therefore, we assume  $T \setminus W\neq \emptyset$ and consider a project $p \in T \setminus W$. In order to show the equivalence of the three notions, it is sufficient to show that $\sat_i(W \cup \{p\}) > \sat_i(T)$ is equivalent to $\sat_i(W) \ge \sat_i(T)$. This is in fact true, as the following chain of equivalences show:
\begin{align*}
    \sat_i(W) \ge \sat_i(T) 
    &\Leftrightarrow \sat(A_i \cap W) \ge \sat(T) \\
    &\Leftrightarrow |A_i \cap W| \ge |T| \\
    &\Leftrightarrow |A_i \cap (W \cup \{p\})|  > |T| \\
    &\Leftrightarrow  |A_i \cap W| + 1 > |T| \\    
    &\Leftrightarrow \sat(A_i \cap (W \cup \{p\})) > \sat(T) \\
    &\Leftrightarrow \sat_i(W \cup \{p\}) > \sat_i(T). 
\end{align*}
Here, we have used $p \in T \subseteq A_i$ and \Cref{lem:1cost}.
\end{proof}

The following example illustrates the difference between \ejrxsat and \ejrosat. 

\begin{example}
Consider one voter and five projects $p_1, p_2, p_3,p_4$ and $p_5$, all approved by this voter. 
The costs and the additive satisfaction function are defined as follows.
\begin{center}
\begin{tabular}{ c|c c c c c }
                  & $p_1$ & $p_2$ & $p_3$ & $p_4$ & $p_5$\\
      \hline
      $c(\cdot)$  & 2.5    & 2.5    &   2.5   &  3    &   4.5 \\
      $\mu(\cdot)$  & 0.1    & 0.1    &   0.1   &  3.1    &   4
      
\end{tabular}
\end{center}
 Let $b = 7$. The single voter is $\{p_1, p_5\}$-cohesive with $\mu(\{p_1, p_5\}) = 4.1$. For this instance, there are three exhaustive outcomes (if one treats $p_1$, $p_2$, and $p_3$ the same).  
 The first one, $\{p_1, p_5\}$, satisfies \ejrsat (and thus also \ejrxsat and \ejrosat). 
The second one, $\{p_2, p_3\}$, violates \ejrxsat since $\mu(\{p_2, p_3\} \cup \{p_1\}) = 0.3 < \mu(\{p_1, p_5\})$; 
 it, however, satisfies \ejrosat since $\mu(\{p_2, p_3\} \cup \{p_5\}) = 4.2 > \mu(\{p_1, p_5\})$. 
 Similarly, $\{p_1, p_4\}$ also satisfies \ejrosat but not \ejrxsat.
\label{ex_diff_x_o}
\end{example}

Having observed that \ejrxsat is strictly stronger than \ejrosat, a natural  question is whether \rx{$\sat$} also satisfies \ejrxsat. This is not the case in general. In \Cref{ex_diff_x_o}, \rx{$\sat$} would first select $p_4$ and then one of $\{p_1, p_2, p_3\}$, and would thus violate \ejrxsat.
However, if we restrict attention to \dns functions $\sat$, we can show that \rx{$\sat$}
always satisfies \ejrxsat.

\begin{restatable}{theorem}{ejrxproof}\label{Rule-X-EJR-x-margUtil}
Let $\sat$ be a \dns function. Then \rx{$\sat$} satisfies \ejrxsat. 
\end{restatable}

\begin{proof}
The proof is similar to the proof of Theorem~\ref{Rule-X-EJR-1}.
As in that proof, let $( A, P, c, b)$ be an ABB instance.
Let $W = \{p_1, \dots, p_k\}$ be the outcome output by \rx{$\sat$}
on this instance where $p_1$ was selected first, $p_2$ second etc.
For any $1 \leq j \leq k$, set $W_j \coloneqq \{p_1, \ldots, p_j\}$.
Consider $N' \subseteq N$, a $T$-cohesive group, for some $T \subseteq P$.
We show that $W$ satisfies \ejrxsat for $N'$.
If $T \subseteq W$ then \ejrxsat is satisfied by definition.
We will thus assume that $T \not \subseteq W$.

In contrast to the proof of Theorem~\ref{Rule-X-EJR-1} we assume this time that
$p^* = \min\{c(p)\mid p \in T \setminus W\}$ is the cheapest project in $T\setminus W$.
Let $k^*$ be the first round, after which there exists
a voter $i^* \in N'$ whose load is larger than $\frac{b}{n} - \frac{1}{c(p^*)}$.
Such a round must exist as otherwise \rx{$\sat$} would not have terminated as 
the voters in $N'$ could still have afforded~$p^*$.
Let $W^* = W_{k^*}$.
Our goal is to prove that there is a voter $i^*$ such that
\begin{align}
	\sat_{i^*}(W^*\cup \{p^*\}) > \sat_{i^*}(T).\label{RuleXEJRxMagUtil:mainclaim1}
\end{align}
Because $\sat$ is additive, the first condition of \dns
implies that if this holds for $p^*$ it must hold for all $p \in T \setminus W$.
Therefore, as $S$ and $T$ were chosen arbitrarily, proving \eqref{RuleXEJRxMagUtil:mainclaim1}
suffices to prove the theorem.

We observe that the derivation of equation \eqref{eq:RuleXEJR1_mainClaim2}
from \eqref{eq:RuleXEJR1_mainClaim} in the proof of Theorem~\ref{Rule-X-EJR-1}
did not depend on the specific choice of $p^*$. Therefore, by the same arguments,
we can prove \eqref{RuleXEJRxMagUtil:mainclaim1} by showing the following:

\begin{equation}
q_{\min} \cdot \sum_{p \in W^* \setminus T} \gamma_{i^*}(p) > q^*_{\max} \cdot \frac{c(T \setminus (W^* \cup \{p^*\}))}{|N'|}.
\label{RuleXEJRxMagUtil:mainclaim2}
\end{equation}
First, we observe that 
$\sum_{p \in W^* \setminus T} \gamma_{i^*}(p) >\frac{c(T \setminus (W^* \cup \{p^*\}))}{|N'|}$
follows by the same argumentation as in the previous proof.
To show that $q_{\min} \geq q^*_{\max}$ we first observe that $q_{min} \geq q^*(p^*)$ holds, because
in every round up to $k^*$ the voters in $N'$ could have paid for $p^*$ on their own,
yet $p^*$ was not selected. Next, we claim that for all $p \in T \setminus W^*$ we have
$q^*(p) \leq q^*(p^*)$. First, observe that by definition
\[q^*(p) = \frac{\sat(p)}{\frac{c(p)}{|N'|}} = |N'|\frac{\sat(p)}{c(p)}.\]
Now, by the choice of $p^*$ we know that $c(p^*) \leq c(p)$ for all $p \in T \setminus W^*$.
Therefore, as $\sat$ is a \dns function,
we have for all $p \in T \setminus W^*$:
\[|N'|\frac{\sat(p^*)}{c(p^*)} \geq |N'|\frac{\sat(p)}{c(p)}.\]

This proves the claim that $q^*(p) \leq q^*(p^*)$ for all $p \in T \setminus W^*$.
From this, we can conclude that $q^*(p^*) = q^*_{\max}$, which means that we have 
$q_{\min} \geq q^*(p^*) = q^*_{\max}$. This concludes the proof of \eqref{RuleXEJRxMagUtil:mainclaim2} and hence the theorem
\end{proof}

This result shows that \rx{$\sat$} is proportional in a strong sense. However, it
also has a big downside: Theorem~\ref{Rule-X-EJR-x-margUtil} only provides a proportionality 
guarantee for \rx{$\sat$} for the specific satisfaction function $\sat$ by which the rule is parameterized.
This means that we have to know which satisfaction function best models the voters when deciding which voting rule to use. It turns out that this is unavoidable, because for two different satisfaction functions, the sets of outcomes providing EJR-x can be non-intersecting. In fact, this even holds for EJR-1. 

\begin{proposition}\label{thm:EJR_Incompatible}
There is an ABB instance for which no outcome satisfies \ejro{$\satCost$} and \ejro{$\satNum$} simultaneously.
\label{thrm:ejr-imposs}
\end{proposition} 
\begin{proof}
Consider the following example with two voters and projects $p_1, \dots, p_{12}$ with $c(p_1)  = c(p_2) = 5$ and the other projects costing $1$. Voter $1$ approves $\{p_1, \dots, p_7\}$ and voter~$2$ approves $\{p_1, p_2, p_8 ,\dots, p_{12}\}$. We set the budget to be $10$. For $\satNum$, we observe that each voter on their own is cohesive over the set of $5$ projects they approve individually (i.e., voter 1 is $\{p_3, \dots, p_7\}$-cohesive and voter 2 is $\{p_8, \dots, p_{12}\}$-cohesive). If either $p_1$ or $p_2$ is included in the outcome, at least one voter has a satisfaction of at most $3$ under $\satNum$; such an outcome can not satisfy \ejro{$\satNum$}. Thus, $W=\{p_3, \dots, p_{12}\}$ is the only outcome satisfying \mbox{\ejro{$\satNum$}}. On the other hand, since both voters together are $\{p_1, p_2\}$-cohesive, the outcome $W$ does not satisfy \mbox{\ejro{$\satCost$}}. Thus, no outcome satisfies both \ejro{$\satCost$} and \ejro{$\satNum$} in this instance.
\end{proof}

\Cref{thrm:ejr-imposs} shows that if we want to achieve strong proportionality
guarantees, we need to know the satisfaction function. Since this might be
unrealistic in practice, in the next chapter we focus on a weaker notion of proportionality.

\section{Proportional Justified Representation}\label{sec:pjr}

In this section, we consider proportionality axioms based on proportional justified representation (PJR).
As our main result in this section,
we show that there exist rules which simultaneously satisfy PJR-x for all \dns functions.
This establishes a counterpoint to our result for EJR at the end of the previous section (\Cref{thrm:ejr-imposs}).

\subsection{Variants of PJR}

PJR is a weakening of EJR. Instead of requiring that, for every cohesive group, there exists a single voter in the group who is sufficiently satisfied, PJR considers the satisfaction generated by the set of all projects that are approved by \textit{some} voter in the group.

    \begin{definition} \label{var:representation:definition:BPJR-L}
    Given an ABB instance $( A, P, c, b )$,
    an outcome $W \subseteq P$ satisfies \emph{PJR with respect to a satisfaction function $\sat$ (\pjrsat)} 
    if and only if for any $T$-cohesive group $N'$ it holds that
        \( \sat((W \cap \bigcup_{i \in N'} A_i)) \ge \sat(T) \text. \)
    \end{definition}

For $\sat=\satCost$, \pjr{$\sat$}  was considered by \citet{aziz:proportionality}, who called it BPJR-L.
For $\sat=\satNum$, \pjr{$\sat$}  was considered by \citet{los:proportional}.

It is straightforward to see that \ejrsat implies \pjrsat. Hence, from \Cref{var:control:EJR:EJR-s-and-BPJR-s-satisfiable} it follows directly that \pjrsat is also always satisfiable.	

    \begin{corollary}\label{var:control:PJR:PJR-s-and-BPJR-s-satisfiable}
         \pjrsat is always satisfiable for any satisfaction function $\sat$.
    \end{corollary}

Since \ejrsat and \pjrsat coincide if there is only one voter,
the hardness proof for \ejrsat (\Cref{thm:NP-EJR}) directly applies to \pjrsat.

\begin{corollary}
Let $\sat$ be a satisfaction function
that is strictly cost-responsive for instances with a single voter.
Then, there is no polynomial-time algorithm that, given an ABB 
instance $( A, P, c, b )$ as input, always computes an outcome
satisfying \pjrsat, unless $P = NP$.
\end{corollary}

The hardness result above (for $\sat=\satCost$) motivated \citet{aziz:proportionality} to define
a relaxation of \pjrsat (for $\sat=\satCost$) they call ``Local-BPJR''.
We discuss this relaxation in 
the full version of this paper,
where we show that it does not imply PJR under the unit-cost assumption. 
\citet{aziz:proportionality} show that their property is satisfied by a polynomial-time computable generalization of the maximin support method \citep{SFFB18a}. 
Instead of Local-BPJR, we consider a stronger property that is similar to \ejrxsat.

    \begin{definition}\label{var:representation:definition:LocalPJR}
        Given an ABB instance $( A, P, c, b )$, an outcome $W$ satisfies \emph{PJR up to any project w.r.t.\ $\sat$} (\pjrxsat) 
        if and only if for any $T$-cohesive group $N'$ and any $p \in T \setminus W$ %
        it holds that
        \( \sat((W \cap \bigcup_{i \in N'} A_i)  \cup \{p\}) > \sat(T) \text. \)
    \end{definition}

Let us consider the relationships between \pjrsat, \pjrxsat and the EJR-based 
fairness notions that we introduced.
By definition, \pjrxsat is implied by \ejrxsat for all satisfaction functions.
One would additionally assume that \pjrxsat is implied by \pjrsat.
Like in the analogous statement for EJR (\Cref{var:representation:proposition:EJR-implies-EJR-x}),
we show this for strictly increasing satisfaction functions.

\begin{restatable}{proposition}{proppjr} \label{var:representation:proposition:PJR-implies-Local-PJR}
Let $\sat$ be a strictly increasing satisfaction function. Then,
\begin{enumerate}
    \item[(i)] \pjrsat implies \pjrxsat, and 
    \item[(ii)] for unit-cost instances if $\sat$ is cost-neutral, \pjrxsat is equivalent to \pjrsat.
\end{enumerate}
\end{restatable}

\begin{proof}
For \textit{(i)}, assume that $W$ satisfies \pjrsat. This means that for any $T$-cohesive group $N'$, we have
           \[ \sat(\bigcup_{i \in N'} A_i \cap W) \geq \sat(T). \]
If $T \subseteq W$, then the requirement of \pjrxsat is satisfied. If not, let $p \in T \setminus W$. By strict increasingness, we have
            \[ \sat(\bigcup_{i \in N'} A_i \cap W \cup \{p\}) > \sat(\bigcup_{i \in N'} A_i \cap W) \geq \sat(T), \]
and thus \pjrxsat holds.

For \textit{(ii)}, we need to show that \pjrxsat implies \pjrsat. 
Consider an outcome $W\subseteq P$ satisfying \pjrxsat and a $T$-cohesive group $N'$. If $T \subseteq W$, then $\sat((W \cap \bigcup_{i \in N'} A_i)) \ge \sat(T)$ trivially holds. Therefore, assume that $T \setminus W \neq \emptyset$ and consider $p \in T \setminus W$.
Using \Cref{lem:1cost}, we get
\begin{align*}
    &\sat((W \cap \bigcup_{i \in N'} A_i) \cup \{p\}) >\sat(T) \\
    &\Leftrightarrow |(W \cap \bigcup_{i \in N'} A_i) \cup \{p\}| > |T| \\
    &\Leftrightarrow |W \cap \bigcup_{i \in N'} A_i| + 1 > |T| \\
    &\Leftrightarrow |W \cap \bigcup_{i \in N'} A_i| \ge |T| \\
    &\Leftrightarrow \sat(W \cap \bigcup_{i \in N'} A_i) \ge \sat(T).
\end{align*}
Since $N'$ was chosen arbitrarily, this implies that $W$ satisfies \pjrsat. 
\end{proof}

For unit-cost instances and cost-neutral and strictly increasing satisfaction functions, the second part of \Cref{var:representation:proposition:PJR-implies-Local-PJR} implies that \pjrxsat is equivalent to the original definition of PJR \citep{fernandez:proportionalJustified}.\footnote{According to this definition, an outcome $W$ satisfies EJR if $|W \cap \bigcup_{i \in N'} A_i| \ge \ell$ for every $\ell$-cohesive group $N'$.} 
Under these conditions, the equivalence also holds for the following weakening of \pjrxsat, which was considered by \citet{los:proportional} for $\sat=\satNum$. 
\label{app:pjr1}

    \begin{definition}\label{var:representation:definition:PJR-1}
        An outcome $W$ satisfies Proportional Justified Representation up to one project (\pjrosat) with respect to an ABB instance $( A, P, c, b )$ if and only if for any $T$-cohesive group $S$ either $T \subseteq W$ or there exists a $p \in \bigcap_{v \in S} A_v \setminus W$ such that
        \[ \sat((W \cap \bigcup_{v \in S} A_v)  \cup \{p\}) > \sat(T) \]

    \end{definition}

We say more about \pjrosat in \Cref{app:bpjr}.

Next, we consider the relationship between \pjrxsat and
\ejrosat. Of course, EJR is generally a stronger axiom that PJR.
However, ``up to one project'' is a greater weakening than ``up to any project''
and, indeed, we find that \ejrosat does not imply \pjrxsat in general.
We keep the following example fairly general to show that 
\ejrosat does not imply \pjrxsat 
for a large class of satisfaction functions.

\begin{example}
Consider a strictly increasing satisfaction function $\sat$
and an ABB instance $( A, P, c, b)$ with $|P| \geq 3$, and one voter $1$ who approves all projects in $P$. Moreover, assume that there is a project $p_1 \in P$ for which  
\[  c(P \setminus \{p_1\}) \leq c(p_1) \quad \text{and} \quad \sat(P \setminus \{p_1\}) = \sat(p_1).\]
Finally, let $b = c(p_1)$. For example, for $\sat \in \{\satCost, \sat^{\text{share}}\}$, we can use 
any example for which $c(P \setminus \{p_1\}) = c(p_1)$.

Let $p_2 \in P$ with $p_2 \neq p_1$ and $P^* = P \setminus \{p_1,p_2\}$.
Since $|P| \geq 3$ we have that $P^* \neq \emptyset$.
We claim that $P^*$ satisfies \ejrosat but not \pjrxsat. Let us first consider \ejrosat:
We observe that $\{1\}$ is $\{p_1\}$-cohesive and $\{p_1\}$ is an affordable outcome
from which $1$ derives maximal satisfaction. 
Moreover, as $\sat$ is a satisfaction function and because $P^* \neq \emptyset$,
we know that $\sat(W) > 0$. Since $\sat$ is strictly increasing, this implies
$\sat(p_1) < \sat(P^* \cup \{p_1\})$.
Hence, $P^*$ satisfies \ejrosat.
        
On the other hand, since $1$ derives the same satisfaction
from the outcomes $\{p_1\}$ and $P \setminus \{p_1\}$, we know that $P \setminus \{p_1\}$ is
also an outcome from which the voter derives maximal satisfaction.
By definition, $P \setminus \{p_1\}$ is a proper superset of $P^*$.
Moreover, by assumption $P \subset \{p_1\}$ is within the budget limit.
This means that $P^*$ violates \pjrxsat. 
\end{example}

\subsection{Achieving PJR-x for All \dns Functions}

Next we turn to our main result on PJR. We give a family of voting rules, all of which simultaneously satisfy \pjrxsat for all \dns functions $\sat$.
To define these voting rules, we recall the definition of priceability, which has been introduced in multiwinner voting by \citet{PeSk20a} and
extended to the PB setting by \citet{PPS21a} and \citet{los:proportional}.

\newcommand{\unspent}[1]{B^*_{#1}}
 \begin{definition}[Priceability]
    An outcome $W$ satisfies priceability with respect to an ABB instance $( A, P, c, b )$ if and only if there is a budget $B >0$ and a collection $d = (d_i)_{i \in N}$ of payment functions $d_i \colon P \to [0,\frac{B}{n}]$ such that\footnote{The numbering of constraints follows \citet{peters:market}.}
    \begin{enumerate}[left=0.5em]
    \setlength\itemsep{0.5em}
        \item[\textbf{C1}] If $d_i(p_j) > 0$ then $p_j \in A_i$ for all $p_j \in P$ and $i \in N$
        \item[\textbf{C2}] If $d_i(p_j) > 0$ then $p_j \in W$ for all $p_j \in P$ and $i \in N$
        \item[\textbf{C3}] $\sum_{p_j \in P} d_i(p_j) \le \frac{B}{n}$ for all $i \in N$
        \item[\textbf{C4}] $\sum_{i \in N} d_i(p_j) = c(p_j) $ for all $p_j \in W$
        \item[\textbf{C5}] $\sum_{i \in N_j} \unspent{i} \le c(p_j)$ for all $p_j \notin W$,
        where $\unspent{i}$ is the unspent budget of voter $i$, i.e.,   $\unspent{i} = \frac{B}{n} - \sum_{p_k \in P}d_i(p_k)$.

    \end{enumerate}
    The pair $\{B,d\}$ is called a \emph{price system} for $W$. 
    \end{definition}
    For unit-cost instances, every exhaustive, priceable outcome satisfies PJR \citep{PeSk20a}. For $\satCost$, we show something similar in the approval-based PB setting. %
    
    \begin{restatable}{theorem}{satcostprice}\label{thm:exhaustive}
    Let $W$ be an outcome such that there is a price system $\{B,d\}$ with $B > b$.
    Then $W$ fulfills %
    \pjrx{$\satCost$}.
    \end{restatable}

    \begin{proof}
    Assume that $W$ does not satisfy \pjrx{$\satCost$}. Then there is a $T$-cohesive group of voters $N'$ and a $p \in \bigcap_{i \in N'} A_i \setminus W$ such that
    \begin{align*} c(W \cap \bigcup_{i \in N'} A_i)) + c(p) = \satCost((W \cap \bigcup_{i \in N'} A_i)  \cup \{p\})\\ \le \satCost(T) \le \frac{\lvert N'\rvert b}{n}. \end{align*}
    We know that $B > b$. Therefore, we get that
    \begin{align*}
    \sum_{i \in N_p} \unspent{i} \ge  \sum_{i \in N'} \unspent{i} \ge \frac{\lvert N' \rvert B}{n} - \sum_{i \in N'} \left(\sum_{p' \in P}d_i(c')\right)\\ \ge \frac{\lvert N' \rvert B}{n} - \satCost(W \cap \bigcup_{i \in N'} A_i) \\
    \ge \frac{\lvert N' \rvert B}{n} - \satCost(T) + c(p)\\\ > \frac{\lvert N' \rvert b}{n} - \frac{\lvert N' \rvert b}{n} + c(p) = c(p)
    \end{align*}
    This is a contradiction to axiom \textbf{C5}. Hence, $W$ must satisfy \pjrx{$\satCost$}.
 \end{proof}

However, %
this implication does not hold for other 
satisfaction functions, as the following example illustrates. 
\begin{example}
\label{ex:pricenoimp}
Consider $\satNum$ and an instance with two voters, five projects $p_1, \dots, p_5$, and budget $b = 4$. The voters have the approval sets $A_1 = \{p_1, p_2, p_3\}$ and $A_2 = \{p_1, p_4, p_5\}$. The project $p_1$ costs $4$ while the rest of the projects cost $1$ each. Then the outcome $\{p_1\}$ is priceable with a budget of $B = 4.5 > 4$ (with both voters paying $2$ for $p_1$), but does not satisfy \pjrx{$\satNum$}. 
\end{example}

\noindent
Towards a more broadly applicable variant of \Cref{thm:exhaustive}, we introduce a new constraint for price systems:
    \begin{enumerate}[left=0.5em]
        \item[\textbf{C6}] $\sum_{i \in N_j} d_i(p_k) \le c(p_j)$ for all $p_j \notin W$ and all $p_k \in W$.
    \end{enumerate}
    Intuitively, a violation of this axiom would mean that the approvers of $p_j$ could take their money they spent on $p_k$ and buy $p_j$ instead for a strictly smaller cost.
    If an outcome is priceable with a price system satisfying \textbf{C6}, we say that it is \textbf{C6}-priceable.  For instance, in \Cref{ex:pricenoimp}, the outcome consisting only of $p_1$ is not \textbf{C6}-priceable since at least one voter must spend at least $2$ on $p_1$ which is more than the price of one of $\{p_2, \dots, p_5\}$. %

    Using this definition, we can now show our main result, namely that \textbf{C6}-priceability with $B > b$ is sufficient for satisfying \pjrxsat for all \dns functions $\sat$.
    
    \begin{restatable}{theorem}{exhaustpjrx} \label{thm:exdnspjr}
    Let $W \subseteq P$ be a \textbf{C6}-priceable outcome with price system $\{B, d\}$ such that \mbox{$B > b$}. Then, $W$ satisfies \pjrxsat for all \dns functions~$\sat$.
    \end{restatable}
       \begin{proof}
For the sake of a contradiction, assume that $W$ does not satisfy \pjrxsat . Then there is a $T$-cohesive group of voters $N'$ and some $p \in T \setminus W$ such that
\begin{equation}\label{eq:main}
\sat((W \cap \bigcup_{i \in N'} A_i)  \cup \{p\}) \le \sat(T). 
\end{equation}
For ease of notation, let $W' \coloneqq W \cap \bigcup_{i \in N'} A_i$ 
be the set of projects in $W$ that are approved by at least one voter in $N'$.
Furthermore, we let $N_p$ denote the set of approvers of $p$.

The proof proceeds in two parts. First, we show that if the voters in $N'$
would additionally buy $p$, then they would spend more than $c(T)$.
To prove this, we mainly use the priceability of $W$. 
Second, we show that there is an unchosen project in $T$ which would give the voters in $N'$ a better satisfaction-to-cost ratio. %
For this part, \textbf{C6} will be crucial, as it guarantees that cheaper projects are
bought first; since $\sat$ is a \dns function, this leads to a higher satisfaction per cost. 
Together, these two parts contradict \eqref{eq:main}.

\smallskip

For the first part, we want to show the following claim:
\begin{equation}
  c(p) + \sum_{i \in N'} \sum_{p' \in W'}  d_i(p') > c(T).  \label{eq:claim}
\end{equation}
Since $B > b$, we obtain from $\textbf{C5}$ that
\[ c(p) \ge \sum_{i' \in N'} \frac{B}{n} - \sum_{p' \in P} d_i(p')
= \frac{|N'|B}{n} - \sum_{p' \in W'} \sum_{i \in N'} d_i(p').\] 
Rewriting this inequality gives us
\[
c(p) + \sum_{p' \in W'} \sum_{i \in N'} d_i(p')  \ge \frac{\lvert N'\rvert B}{n} > \frac{\lvert N'\rvert b}{n} \ge c(T).
\]

Having shown \eqref{eq:claim}, we now advance to the second part of the proof. Here we want to compare the satisfaction per unit of money between $W' \cup \{p\}$ and $T$.
Since both the satisfaction function $\sat$ and the cost function $c$ are additive, we can ignore the projects that appear
both in $W' \cup \{p\}$ and $T$ when doing so. Let $T_W = T \cap W'$. Then,
we first observe that \eqref{eq:main} implies by the additivity of $\sat$ that
\begin{equation}\label{eq:sat-diff}
\sat(W'\setminus T_W) \le \sat(T\setminus(T_W\cup \{p\})).
\end{equation}
We apply the same idea to \eqref{eq:claim}. Since for all $p' \in W'$ it holds that $\sum_{i \in N'} d_i(p') \le c(p')$ we get that
\begin{equation}\label{eq:cost-diff}
\sum_{i \in N'} \sum_{p' \in W'\setminus T_W}  d_i(p') > c(T\setminus (T_W \cup \{p\})).
\end{equation}
We now show that $T\setminus (T_W \cup \{p\}) \neq \emptyset$. Assume for contradiction that $T\setminus (T_W \cup \{p\}) = \emptyset$, then
$\sat(T\setminus(T_W\cup \{p\})) =0$.
By \eqref{eq:sat-diff} this implies $\sat(W' \setminus T_W) =0$
and hence $W'\setminus T_W = \emptyset$ 
Then, however, both sides of \eqref{eq:cost-diff} evaluate 
to $0$; a contradiction. Thus, we know that $c(T\setminus (T_W \cup \{p\})) > 0$.

By putting \eqref{eq:sat-diff} and \eqref{eq:cost-diff} together, we get that
\begin{align*}
    \frac{\sat(W' \setminus T_W )}{\sum_{p' \in W'\setminus T_W} \sum_{i \in N'} d_i(p')} < \frac{\sat(T\setminus (T_W \cup \{p\}))}{c(T\setminus (T_W \cup \{p\}))}.
\end{align*}
Since $\sat$ and $c$ are additive, we can rewrite this inequality as
\begin{align*}
   \sum_{p' \in W'\setminus T_W}  \frac{\sat(p')}{\sum_{i \in N'} d_i(p') } 
    <\sum_{t \in T\setminus (T_W \cup \{p\})} \frac{ \sat(t)}{ c(t)}.
\end{align*}
Now we use the fact that $\min(\frac{a}{c}, \frac{b}{d}) \le \frac{a+b}{c+d} \le \max(\frac{a}{c}, \frac{b}{d})$ to obtain the following:
\begin{align*}
    \min_{p' \in W'\setminus T_W}\left\{\frac{\sat(p')}{ \sum_{i \in N'} d_i(p')}\right\} \le
   \sum_{p' \in W'\setminus T_W}\frac{ \sat(p')}{\sum_{i \in N'} d_i(p') } \\< \sum_{t \in T\setminus (T_W \cup \{p\})}\frac{ \sat(t)}{ c(t)} \le \max_{t \in T\setminus (T_W \cup \{p\})}\left\{\frac{\mu(t)}{c(t)}\right\}.
\end{align*}
Let $p_{\min} = \argmin_{p' \in W'\setminus T_W}\left\{\frac{\sat(p')}{\sum_{i \in N'} d_i(p')}\right\}$
and $t_{\max} = \argmax_{t \in T\setminus T_W}\left\{\frac{\mu(t)}{c(t)}\right\}$. 
Then it follows that 
\begin{equation}\label{eq:normUtil}
\frac{\sat(p_{\min})}{c(p_{\min})} \le \frac{\sat(p_{\min})}{\sum_{i \in N'} d_i(p_{\min})} < \frac{\sat(t_{\max})}{c(t_{\max})}\text.
\end{equation}
In other words, $p_{\min}$ has a lower normalized satisfaction than $t_{\max}$.
Since $\sat$ is a \dns function, we can conclude that $c(t_{\max}) \le c(p_{\min})$.
By the first condition of \dns functions, this implies
$\sat(p_{\min}) \ge \sat(t_{\max})$. However, then for the second inequality of
\eqref{eq:normUtil} to hold, we must have $\sum_{i \in N'} d_i(p_{\min}) > c(t_{\max})$,
a contradiction to \textbf{C6}.
\end{proof}

    First, we observe that from the MES family of rules \rx{$\satNum$} satisfies the conditions of the theorem.
    \begin{restatable}{corollary}{rulexpjr}
\rx{$\satNum$} satisfies \pjrxsat for all \dns functions $\sat$.
\end{restatable}

\begin{proof}
For this, it is sufficient to show that \rx{$\satNum$} always returns an outcome that is \textbf{C6}-priceable for a $B > b$.
We observe that requirements \textbf{C1} to \textbf{C5}, are naturally satisfied by the price system constructed throughout \rx{$\satNum$}. The requirement that $B > b$ is however not naturally satisfied. To change this, let $\delta = \min_{p \in P} c(p)-(\sum_{i \in N_p} B^*_i) $. Since, no further project is affordable, we know that $\delta > 0$. We now set $B_\delta = B + \frac{\delta}{n}$ to be the new budget. Requirements \textbf{C1} to \textbf{C4} still naturally hold for this budget. Further, we know that $c(p)-(\sum_{i \in N_p} B^*_i) \ge \delta$ and hence $c(p) \ge \sum_{i \in N_p} (B^*_i + \frac{\delta}{n})$. Hence, \textbf{C5} is also satisfied and we only need to show that \rx{$\satNum$} indeed satisfies \textbf{C6}. Hence, we need to show that for any $p_j \notin W$ and $p_k \in W$ it holds that $\sum_{i \in N_j} d_i(p_k) \le c(p_j)$. Assume on the contrary that $\sum_{i \in N_j} d_i(p_k) > c(p_j)$. At the time $p_k$ gets bought we know that $p_k$ is $\rho$ affordable, while $p_j$ is not $\rho'$ affordable for any $\rho' < \rho$ and thus $\sum_{i \in N_j} \min(b_i, \rho') < c(p_j)$. Further, we know that $\min(b_i, \rho) = d_i(p_k)$ for any $i \in N'$. 

Thus, we know that \begin{align*}
    c(p_j) < \sum_{i \in N_j} d_i(p_k) = \sum_{i \in N_j \cap N_k} \min(b_i, \rho)\\  \le  \sum_{i \in N_j} \min(b_i, \rho).
\end{align*} Let $N_{\min} = \{i \in N_j \colon b_i < \rho\}$. If $N_{\min} = N'$, we could set $\rho' = \max_{i \in N'} b_i < \rho$ and would thus get 
\begin{align*}
    c(p_j) < \sum_{i \in N_j} \min(b_i, \rho) = \sum_{i \in N_j} \min(b_i, \rho') < c(p_j)
\end{align*} and thus a contradiction. Otherwise, we could pick $\rho'$ such that $\lvert N_j \rvert (\rho - \rho') < (\sum_{i \in N_j} \min(b_i, \rho) - c(p_j))$. Then since $p_j$ is not $\rho$ affordable get that \begin{align*}
    c(p_j) &> \sum_{i \in N_j} \min(b_i, \rho') \\ &\ge \sum_{i \in N_j} \min(b_i, \rho) + \lvert N_j \rvert (\rho' - \rho)\\
    &> c(p_j).
\end{align*}
\end{proof}

Next, we present an example showing that \rx{$\satCost$} does not satisfy \textbf{C6}. Similar counterexamples can be constructed for other satisfaction function $\sat \neq \satNum$. As a consequence, \Cref{thm:exdnspjr} does not apply to those variants of MES.

 \begin{example}
 \label{exp:dnsmes}
 Consider an ABB instance $( A, P, c, b )$
 with two voters, three projects and budget $b = 3$,
 where project $p_1$ costs $3$ and is approved by both voters,
 project $p_2$ costs $1$ and is only approved by voter $1$, and
 project $p_3$ costs $1$ and is only approved by voter $2$.

 In this example, the outcome of \rx{$\satCost$} is $\{p_1\}$. Assume that there exists
 a budget $B$ and payment functions $d_1$ and $d_2$ such that \textbf{C1}-\textbf{C6} 
 are satisfied. Then, $d_1(p_1) + d_2(p_1) = 3$. Hence either $d_1(p_1)$ or $d_2(p_1)$
 must be larger than $1$. Assume w.l.o.g.\ $d_1(p_1) > 1$. Then we have 
 \(\sum_{i \in N_2} d_i(p_1) = d_1(p_1) > 1 = c(p_2).\) 
 This contradicts \textbf{C6}.
 \end{example}

    Two further rules for which we can always find such a price system are the PB versions of sequential Phragm\'{e}n \citep{phragmen:sur,brill:phragmen} and the maximin support method \citep{SFFB18a}.
    For the definitions of these two rules, we refer to the 
    Appendix~\ref{app:phrag}. 
    \begin{restatable}{corollary}{corphragmms}\label{cor:phrag}
    Sequential Phragm\'{e}n and the maximin support method provide \pjrxsat for all \dns functions $\sat$.
    \end{restatable}

\Cref{cor:phrag} is proved in \Cref{app:phrag}.

Finally, we can show that \dns is, in a sense, a necessary restriction. Namely, we can show that for any function mapping costs to satisfaction in a way that violates DNS, we can find an instance such that \rx{$\satNum$} does not satisfy \mbox{PJR-x} for that instance. We give an informal statement of the theorem here and a full statement and proof in 
the full version of this paper. 

\begin{restatable}{proposition}{pjrdnsclass} \label{prop:dns-necessary}
Let $\mu$ be an additive satisfaction function that is not a \dns function. Then there exists an ABB instance $(A,P,c,b)$ with satisfaction function $\sat$ such that \rx{$\satNum$  } violates \pjrx{$\sat$}.
\end{restatable}

\begin{proof}

We formally want to show that for any function $s\colon \mathbb{R}^+ \to \mathbb{R}^+$, which is not DNS in the sense that there are either values $x, x' \in \mathbb{R}$ with $x \le x'$ such that 
\begin{enumerate}
    \item[(i)] $\sat(x) > \sat(x') $ or
    \item[(ii)] $\frac{\sat(x)}{x} < \frac{\sat(x')}{x'} $.
\end{enumerate}
there is some PB instance, with $\sat(p) = s(c(p))$ such that \rx{$\satNum$} does not satisfy \pjrxsat. For readability we will just write $\sat$ instead of $s$.

Since $\sat$ is not DNS, there are values $x, x' \in \mathbb{R}$ with $x \le x'$ such that 
\begin{enumerate}
    \item[(i)] $\sat(x) > \sat(x') $ or
    \item[(ii)] $\frac{\sat(x)}{x} < \frac{\sat(x')}{x'} $.
\end{enumerate}
\paragraph{Case (i).} If $x \le x'$ but $\sat(x) > \sat(x')$. Without loss of generality, we can scale the instance such that $x = 1$ and $\sat(x) = 1$.  Let $\frac{p}{q} \in \mathbb{Q} \cap [x' - 1, x' - 1 + \varepsilon]$ for a sufficiently small $\varepsilon$ (we will see  later what ``sufficiently small'' means). Next, since $\sat(x') < 1$ we can choose $\beta \in \mathbb{N}$ large enough such that \begin{align}
    \sat(x') + \frac{1}{\beta} < 1 \label{eq:bound}
\end{align}
We set the budget to $b = \beta (x' + \varepsilon)$.
There are $2\beta + 2$ projects in total; half of them cost $1$ and half of them cost~$x'$.
There are $q+p$ voters; $q$ voters approve all projects and $p$ voters approve only the projects of cost~$x'$.

The $q$ voters approving everything are cohesive over $\beta$ of the projects of cost $1$, because
\begin{align*}\frac{bq}{p + q} = \frac{(\beta (x' + \varepsilon)) q }{p + q} \ge \frac{(\beta (x' + \varepsilon)) q }{(x' - 1 + \varepsilon)q + q}\\ = \frac{\beta (x' + \varepsilon) }{(x' + \varepsilon)} = \beta  \end{align*}

Let us consider how \rx{$\satNum$} behaves on this instance. Since \begin{align*}
    \frac{qx'}{p + q} \le \frac{q x'}{q + q(x'-1)} = 1 ,
\end{align*}
we have $\frac{1}{q} \ge \frac{x'}{p+q}$. Therefore, \rx{$\satNum$} selects projects of cost $x'$ until the whole budget is used. We can choose $\varepsilon$ small enough, such that $\beta\varepsilon < 1$. Then, \rx{$\satNum$} would choose exactly $\beta$ of the projects of cost $x'$ after which only $\beta\varepsilon$ budget is left, which is not enough to afford any other project. Following this, all voters have a satisfaction of $\beta \sat(x')$. Thus, using \eqref{eq:bound}, we obtain that $\beta \sat(x') + 1 < \beta$ and therefore \pjrxsat is not satisfied.

\paragraph{Case (ii).} Next, we assume that $x \le x'$ but $\frac{\sat(x)}{x} < \frac{\sat(x')}{x'}$. In this case we only need a single voter for whom \rx{$\satNum$} only buys projects of cost $x$, while this single voter values projects of cost $x'$ more. We again scale the instance such that $x = 1$ and $\sat(x) = 1$. Thus, we know that $\frac{\sat(x')}{x'} > 1$. Due to this, there must exist a $\beta \in \mathbb{N}$ with $\frac{x'}{\sat(x')} < \frac{\beta - 1}{\beta}$. We set the budget $b = \beta x'$. 

It is easy to see that the single voter is cohesive over~$\beta$ projects of cost $x'$ with a utility of $\beta \sat(x')$. However, \rx{$\satNum$} would instead buy at most $b$ projects of cost $1$, resulting in a satisfaction of $b = \beta x'$. Since $\frac{x'}{\sat(x')} < \frac{\beta - 1}{\beta}$, we thus obtain $\beta x' + \sat(x') < \beta \sat(x')$, and hence \pjrxsat is not satisfied.
\end{proof}

\section{Conclusion}
\label{sec:conclusion}

We have studied proportionality axioms for participatory budgeting elections based on approval ballots. Our results can be summarized along two main threads:

\begin{enumerate}
\item
If strong (i.e., EJR-like) proportionality guarantees are desired, then it is necessary to know the satisfaction function, as different satisfaction functions may lead to incompatible requirements (\Cref{thrm:ejr-imposs}). If the satisfaction function is known and belongs to the class of DNS functions, however, 
we can guarantee \textit{EJR up to any project} using a polynomial-time computable variant of MES tailored to this function (\Cref{Rule-X-EJR-x-margUtil}).
   
 \item 
If the proportionality requirement is weakened to a PJR-like notion, there is no need to know the satisfaction function precisely: We identify a large class of satisfaction functions so that \textit{PJR up to any project} is achievable for all those functions simultaneously (Theorem~\ref{thm:exdnspjr}).
We identify a class of voting rules that achieve this, including Phragm\'{e}n's sequential rule, the maximin support method, and a variant of MES. (Among those three rules, the MES variant is the only rule that additionally satisfies EJR w.r.t. %
the cardinality-based satisfaction function.)

\end{enumerate}

It is open whether we can even achieve EJR-x (or even \mbox{PJR-x}) in polynomial time for additive non-DNS functions. Here, it seems crucial to further identify rules --- besides MES --- providing proportionality guarantees for PB. Furthermore, it would be interesting to push the boundaries of \Cref{thm:exdnspjr}; for example, can we soften the assumption that we use the same satisfaction function for all voters?

It is also an open question whether proportional outcomes can be computed in polynomial time for satisfaction functions that are not additive (e.g., for submodular or subadditive satisfaction functions).
Looking beyond the approval-based setting, it would be interesting to extend our framework to general (additive or non-additive) utility functions.

\section*{Acknowledgements}

We thank Dominik Peters for helpful comments. 
This work was supported by the Austrian Science Fund (FWF) under the research grants P31890 and J4581 and by the Deutsche Forschungsgemeinschaft (DFG) under the grant BR~4744/2\nobreakdash-1 and the Graduiertenkolleg ``Facets of Complexity'' (GRK~2434).

\bibliographystyle{apalike}
\bibliography{literature}

\clearpage
\appendix

\section{A strengthening of $\sat$-EJR-1}\label{App:EJR-1}

In this section, we consider a slightly stronger version of \ejrosat
that we call \ejrosatplus. The only difference to \ejrosat is that 
we require the project $p$ to come from $T$.

    \begin{definition}\label{var:representation:definition:EJR-1+}
An outcome $W$ satisfies \ejrosatplus with respect to an
satisfaction function $\sat$ and an ABB
instance $( A, P, c, b )$ if and only if for every
$T$-cohesive group $N'$ either $T \subseteq W$
or there exists a voter $i \in N'$ and a project $p \in T \setminus W$ such that 
	\[\sat_i(W \cup \{p\}) > \sat_i(T).\]

    \end{definition}

By definition, \ejrosatplus implies \ejrosat and is implied by \ejrxsat.
We show that \rx{$\sat$} satisfies \ejrosatplus for additive satisfaction functions.
The proof of this result will also serve as a blueprint for proving Theorem~\ref{Rule-X-EJR-x-margUtil}

\begin{restatable}{theorem}{ejrorx}\label{Rule-X-EJR-1}
Let $\sat$ be an additive satisfaction function. Then \rx{$\sat$}
satisfies \ejrosatplus.
\end{restatable}

\begin{proof}
Let $( A, P, c, b)$ be an ABB instance.
Let $W = \{p_1, \dots, p_k\}$ be the outcome output by \rx{$\sat$}
on this instance where $p_1$ was selected first, $p_2$ second etc.
For any $1 \leq j \leq k$, set $W_j \coloneqq \{p_1, \ldots, p_j\}$.
Consider $N' \subseteq N$, a $T$-cohesive group, for some $T \subseteq P$.
We show that $W$ satisfies \ejrosatplus for $N'$.
If $T \subseteq W$, then \ejrosatplus is satisfied by definition.
We will thus assume that $T \not \subseteq W$.

Let $p^* = \max\{c(p)\mid p \in T \setminus W\}$ the most expensive project
in $T$ that is not in $W$. Let $k^*$ be the first round after which there exists
a voter $i^* \in S$ whose load is larger than $\frac{b}{n} - \frac{1}{c(p^*)}$.
Such a round must exist as otherwise \rx{$\sat$} would not have terminated as 
the voters in $N'$ could still have afforded $p^*$.
Let $W^* = W_{k^*}$.
Our goal is to prove that $W^*$ satisfies \ejrosatplus for $S$
as there is a voter $i^*$ such that
\begin{align}
	\sat_{i^*}(W^*\cup \{p^*\}) > \sat_{i^*}(T)\nonumber
\end{align}
Due to the additivity of $\sat$ this is equivalent to
\begin{align}
	\Leftrightarrow \quad & \sat_{i^*}(W^*) > \sat_{i^*}(T \setminus \{p^*\}) \nonumber \\
	\Leftrightarrow \quad & \sat_{i^*}(W^* \cap T) + \sat_{i^*}(W^* \setminus T) > \nonumber \\
	& ~~\sat_{i^*}(T \cap W^*) + \sat_{i^*}(T \setminus (W^* \cup \{p^*\})) \nonumber \\
	\Leftrightarrow \quad & \sat_{i^*}(W^* \setminus T) > \sat_{i^*}(T \setminus (W^* \cup \{p^*\})). \label{eq:RuleXEJR1_mainClaim}
\end{align}
We now work on each side of inequality \eqref{eq:RuleXEJR1_mainClaim} to eventually prove that it is indeed satisfied.

\medskip

We start by the left-hand side of \eqref{eq:RuleXEJR1_mainClaim}.
Let us first introduce some notation that allows us to reason in terms of satisfaction
per unit of load. For a project $p \in W$, we denote by $\alpha(p)$ the smallest $\alpha \in \mathbb{Q}_{> 0}$ such that $p$ was $\alpha$-affordable when \rx{$\sat$} selected it.
Moreover, we define $q(p)$\,---\,the satisfaction that a voter that contributes fully to $p$
gets per unit of load\,---\,as $q(p) \coloneqq \frac{1}{\alpha(p)}$.

Since before round $k^*$, voter $i^*$ contributed in full for all projects in $W^*$
(as $\ell_{i^*} < \frac{b}{|S|}$ after each round $1, \dots, k^*$),
we know that $\alpha(p) \cdot \sat_{i^*}(\{p\})$
equals the contribution of $i^*$ for $p$ for any $p \in W^*$. We thus have:
\begin{align}
	& \sat_{i^*}(W^* \setminus T) \nonumber \\
	= \quad & \sum_{p \in W^* \setminus T} \sat_{i^*}(\{p\}) \nonumber \\
	= \quad & \sum_{p \in W^* \setminus T} \alpha(p) \cdot \sat_{i^*}(\{p\}) \cdot \frac{1}{\alpha(p)} \nonumber \\
	= \quad & \sum_{p \in W^* \setminus T} \gamma_{i^*}(p) \cdot q(p), \label{eq:RuleXEJR1_LeftHand1}
\end{align}
where $\gamma_{i^*}(p)$ denotes the contribution of $i^*$ to any $p \in W$,
defined such that if $p$ has been selected at round $j$, i.e., $p = p_j$, then 
$\gamma_{i^*}(p) = \gamma_{i^*}(W_j, \alpha(p_j), p_j)$.

Now, let us denote by $q_{\min}$ the smallest $q(p)$ for any $p \in W^* \setminus T$. From \eqref{eq:RuleXEJR1_LeftHand1}, we get
\begin{equation}
	\sat_{i^*}(W^* \setminus T) \geq q_{\min} \sum_{p \in W^* \setminus T} \gamma_{i^*}(p). \label{eq:RuleXEJR1_LeftHand2}
\end{equation}

\medskip

We now turn to the right-hand side of \eqref{eq:RuleXEJR1_mainClaim}. We introduce some additional notation for that. For every project $p \in T$, we denote by $q^*(p)$ the share per load that a voter
in $S$ receives if only voters in $N'$ contribute to $p$, and they all contribute in full to $p$, defined as
\[q^*(p) = \frac{\sat_v(p)}{\frac{c(p)}{|N'|}},\]
where $v$ is any voter in $N'$.

We have
\begin{align}
	& \sat_{i^*}(T \setminus (W^* \cup \{p^*\})) \nonumber \\
	= \quad & \sum_{p \in T \setminus (W^* \cup \{p^*\})} \sat_{i^*}(p) \nonumber \\
	= \quad & \sum_{p \in T \setminus (W^* \cup \{p^*\})}
	\frac{\sat_{i^*}(p)}{\frac{c(p)}{|N'|}} \cdot \frac{c(p)}{|S|} \nonumber \\
	= \quad & \sum_{p \in T \setminus (W^* \cup \{p^*\})} q^*(p) \cdot \frac{c(p)}{|N'|} \label{eq:RuleXEJR1_RightHand1}
\end{align}
Setting $q^*_{\max}$ to be the largest $q^*(p)$ for all $p \in T \setminus (W^* \cup \{p^*\})$, \eqref{eq:RuleXEJR1_RightHand1} gives us:

\begin{equation}\sat_{i^*}(T \setminus (W^* \cup \{p^*\})) \leq q^*_{\max} \cdot \frac{c(T \setminus (W^* \cup \{p^*\}))}{|N'|}. \label{eq:RuleXEJR1_RightHand2}
\end{equation}

\medskip

In the aim of proving inequality \eqref{eq:RuleXEJR1_mainClaim}, we want to show that 
\begin{equation}
q_{\min} \cdot \sum_{p \in W^* \setminus T} \gamma_{i^*}(p) > q^*_{\max} \cdot \frac{c(T \setminus (W^* \cup \{p^*\}))}{|N'|}.
\label{eq:RuleXEJR1_mainClaim2}
\end{equation}
Note that proving that this inequality holds, would in turn prove \eqref{eq:RuleXEJR1_mainClaim} thanks to
\eqref{eq:RuleXEJR1_LeftHand2} and \eqref{eq:RuleXEJR1_RightHand2}. We divide the proof of \eqref{eq:RuleXEJR1_mainClaim2} into two claims.

\begin{claim}
$q_{\min} \geq q^*_{\max}$.
\end{claim}

\begin{claimproof}
Consider any project $p' \in T \setminus (W^* \cup \{p^*\})$. It must be the case that $p'$ was at least $\frac{1}{q^*(p)}$-affordable in round $1, \dots, k^*$, for all $p \in W^*$,
as all voters in $N'$ could have fully contributed to it based on how we defined $k^*$.

Since no $p' \in T \setminus (W^* \cup \{p^*\})$ was selected by \rx{$\sat$}, we know that all projects
that have been selected must have been at least as affordable, i.e., for all $p \in W^*$ and $p' \in T \setminus (W^* \cup \{p^*\})$ we have:
	 \begin{align*}
	 & \alpha(p) \leq \frac{1}{q^*(p')} \\
	 \Leftrightarrow \quad & q(p) \geq q^*(p') \\
	 \Leftrightarrow \quad & q_{\min} \geq q^*_{\max}.
	 \end{align*}
This concludes the proof of our first claim.
\end{claimproof}

\begin{claim}
$\displaystyle \sum_{p \in W^* \setminus T} \gamma_{i^*}(p) > \frac{|T \setminus (W^* \cup \{p^*\})|}{|N'|}$.
\end{claim}

		 \begin{claimproof}
From the choice of $k^*$, we know that the load of voter $i^*$ at round $k^*$ is such that
	 \[\ell_{i^*}(W^*) + \frac{c(p^*)}{|N'|} > \frac{b}{n}.\]
On the other hand, since $N'$ is a $T$-cohesive group, we know that
	 \[\frac{c(T)}{|N'|} = \frac{c(T \setminus \{p^*\})}{|N'|}+ \frac{c(p^*)}{|N'|} \leq \frac{b}{n}.\]
Linking these two facts together, we get
	 \[\ell_{i^*}(W^*) > \frac{c(T \setminus \{p^*\})}{|N'|}.\]
By the definition of the load, we thus have:
	 \[\ell_{i^*}(W^*) = \sum_{p_j \in W^*} \gamma_{i^*}(p_j) > \frac{c(T \setminus \{p^*\})}{|N'|}.\]
This is equivalent to
	\begin{multline}
	\sum_{p_j \in T \cap W^*} \gamma_{i^*}(p_j) + \sum_{p_j \in T \setminus W^*} \gamma_{i^*}(p_j) > \\ \frac{c(T \cap W^*)}{|N'|} + \frac{c(T \setminus (W^*\cup \{p^*\}))}{|N'|} \label{eq:RuleXEJR1_SecondClaim}
	\end{multline}
Now, we observe that every voter in $N'$ contributed in full for every 
project in $W^*$. It follows that the contribution of every voter in $N'$ for a
project $p_j \in T \cap W^*$ is smaller or equal to the contribution needed if 
the voters in $N'$ would fund the project by themselves. In other words, for all
$p \in T \cap W^*$ we have
	\[\gamma_{i^*}(p) \leq \frac{c(p)}{|N'|}.\]
It follows that
	\[\sum_{p_j \in T \cap W^*} \gamma_{i^*}(p_j) \leq \frac{c(T \cap W^*)}{|N'|}.\]
For \eqref{eq:RuleXEJR1_SecondClaim} to be satisfied, we must have that
	\[\sum_{p_j \in W^* \setminus T} \gamma_{i^*}(p_j) > \frac{c(T \setminus (W^*\cup \{p^*\}))}{|N'|}\]
This concludes the proof of our second claim.
\end{claimproof}

\smallskip

Putting together these two claims immediately shows that inequality \eqref{eq:RuleXEJR1_mainClaim2} is satisfied, which in turn shows that \eqref{eq:RuleXEJR1_mainClaim} holds. Since $N'$, $T$ and $p^*$ were chosen arbitrarily, this shows that \rx{$\sat$} satisfied \ejrosatplus.
\end{proof}

\bigskip

\section{Sequential Phragm\'{e}n and the Maximin Support Method}
\label{app:phrag}

Finally, we introduce two more priceable rules satisfying \pjrxsat and \textbf{C6}. The first one is a generalization of Phragmén's sequential rule \citep{brill:phragmen}, the second one a generalization of the maximin support method \citep{SFFB18a}.

  \label{var:representation:definition:seqphrag}
For an ABB instance $( A, P, c, b )$ and a vector $l = (l_i)_{i \in N}$, define
    \[t(l, A, p, c) := \frac{c(p)+\sum_{i \in N_{p}} l_i}{|N_{p}|}\]
Generalized Sequential Phragmén, introduced to PB by \citet{los:proportional}, selects all outcomes $W$ that can result from the following iterative procedure:
       \smallskip
\begin{lstlisting}[escapeinside={(*}{*)}, language=Python]    
(*$W$*) := (*$\emptyset$*)
for all (*$i \in N$*):
    (*$l_i$*) := 0
(*$P$*) := (*$\{p \in P \colon \sum_{i \in N} l_i + c(p) \le b \} $*)
while (*$W \setminus P \neq \emptyset$*): 
    if exists (*$p' \in \argmin_{p \in P\backslash W} t(\vec{l}, A, p, \sat)$ s.t.\ $c(W \cup p') > b$*) :
       break
    else 
        for some (*$p' \in \argmin_{p \in P\backslash W} t(\vec{l}, A, p, \sat)$*):
           for all (*$i \in V_{p'}$*):
               (*$l_i$*) := (*$t(l, A, p', \sat)$*)
           (*$W$*) := (*$W \cup \{p\}$*)
return (*$W$*)
\end{lstlisting}       
  
\smallskip
    
 The (generalized) maximin support method can be seen as a variant of (generalized) sequential Phragm\'{e}n in which the loads are rebalanced in each iteration.\footnote{When \citet{aziz:proportionality} allegedly generalized Phragmén's sequential rule, they in fact generalized the maximin support method because they employed rebalancing.}
To define the maximin support method, for a given outcome $W$ we define a collection $l=(l_i)_{i \in N}$ of functions to be a \textit{set of loads for $W$} if $l_i \colon W \to \mathbb{R}^+$ for each $i \in N$ and the following conditions hold: 
\begin{itemize}
    \item $\sum_{i \in N} l_i(p) = c(p)$ for every $p \in W$
    \item $l_i(p) = 0$ if $p \notin A_i$ for every $i \in N$
\end{itemize}
For a given set of loads $l$, we let $s(l)$ denote the maxmial total load of a voter, i.e., $s(l) = \max_{i \in N} \sum_{p \in W} l_i(p)$.

The maximin support method now iteratively works as follows (in line 4, $l$ denotes a set of loads for $W \cup\{p\}$):

\smallskip

\begin{lstlisting}[escapeinside={(*}{*)}, language=Python]    
(*$W$*) := (*$\emptyset$*)
(*$P$*) := (*$\{p \in P \colon \sum_{p' \in W} c(p') + c(p) \le b \}$*)
while (*$P \neq \emptyset$*): 
    if exists (*$p' \in \argmin_{p \in P, \text{ loads }l} s(l)$ s.t.\ $c(W \cup p') > b$*) :
           break
    for some (*$p' \in \argmin_{p \in P, \text{ loads }l} s(l)$*):
        (*$W$*) := (*$W \cup \{p\}$*)
return (*$W$*)
\end{lstlisting}

\smallskip

\corphragmms*
\begin{proof}
Consider Sequential Phragmén first. \citet{los:proportional} have shown that Sequential Phragmén
is priceable. We modify their proof to show that there always exists a price system with $B > b$
that satisfies \textbf{C6}.
Consider an outcome $W$ of Sequential Phragmén. We can assume $W \neq P$ as otherwise \pjrsat
is trivially satisfied. Then, Sequential Phragmén terminated because there was a project $p'$
such that $p' \in \argmin_{p \in P\backslash W} t(\vec{l}, A, p, \sat)$ and $c(W \cup p') > b$.
Consider the load distribution defined by $l_i := t(l, A, p', \sat)$.
We claim that the following is a valid price system for $W$: Let $B = n \cdot \max_{i\in N}(l_i)$.
Furthermore, let $p_1, \dots, p_k$ be an enumeration of $W$ in the order in which the projects
are selected by Sequential Phragmén and let $l_i^j$ be the load of voter $i$ before project $p_j$
is selected. Finally, let $l_i^0 = 0$ for all $i \in N$. Then we define for all $i \in N$
\[d_i(p) = \begin{cases}
l_i^j - l_i^{j-1} & \text{if } p = p_j \text{ for some } j \leq k \\
0 & \text{else.}
\end{cases}
\]
We observe that \textbf{C1}-\textbf{C4} are satisfied by construction, similar to the proof of
\citet{los:proportional}. For \textbf{C5} we first observe that for $p'$ we have by construction.
\begin{multline*}
\sum_{i \in N_{p'}} B_i^* = \sum_{i  \in N_{p'}} t(\vec{l}, A, p', \sat) -
\sum_{p \in A_ i \cap W} d_i(p) =\\ \sum_{i  \in N_{p'}} t(\vec{l}, A, p', \sat) - l_i^k  = c(p')
\end{multline*}
Further, assume that there is a $p'' \in P \setminus W$ such that
$\sum_{i \in N_{p''}} B_i^* > c(p'')$. Then, we have the following:
\begin{align*}
\sum_{i \in N_{p''}} B_i^* &> c(p'')\\
\sum_{i  \in N_{p'}} t(\vec{l}, A, p', \sat) - l_i^k &> c(p'')\\
|N_{p''}| t(\vec{l}, A, p', \sat) - \sum_{i  \in N_{p''}}l_i^k &> c(p'') \\
 t(\vec{l}, A, p', \sat)  &> \frac{c(p'') + \sum_{i  \in N_{p''}}l_i^k}{|N_{p''}|} \\
 t(\vec{l}, A, p', \sat)  &> t(\vec{l}, A, p'', \sat)
\end{align*}
However, this contradicts $p' \in \argmin_{p \in P\backslash W} t(\vec{l}, A, p, \sat)$.
It follows that \textbf{C5} is satisfied.

Now consider \textbf{C6}. For the sake of a contradiction, assume there are $p_j, p_k$ 
such that $p_j \not \in W$, $p_k \in W$ and $\sum_{i \in N_j} d_i(p_k) > c(p_j)$. 
Furthermore, let $l_i^k$ be the load of voter $i$ before $p_k$ was selected. 
and $l_i^{k+1}$ be the load of voter $i$ after $p_k$ was selected. 
Then, we have the following:
\begin{multline*}
t(\vec{l^k}, A, p_j, \sat) = \frac{c(p_j) + \sum_{i  \in N_{p_j}}l_i^k}{|N_{j}|}\\
< \frac{ \sum_{i \in N_j} d_i(p_k)+ \sum_{i  \in N_{p_j}}l_i^k}{|N_{j}|} \\= 
\frac{ \sum_{i \in N_j} (l_i^{k+1} - l_i^{k} )+ \sum_{i  \in N_{j}}l_i^k}{|N_{j}|} \\=
\frac{ \sum_{i \in N_j} l_i^{k+1}}{|N_{j}|} =%
t(\vec{l^k}, A, p_k, \sat)
\end{multline*}
This is a contradiction to $p_k \in \argmin_{p \in P\backslash W} t(\vec{l}, A, p, \sat)$
in the round where $p_k$ is selected.

It remains to show that $B > b$. We know that 
\[\sum_{i\in N_{p'}} \left(\frac{B}{|N|} -  \sum_{p \in P} d_i(p)\right) = c(p')\]
Moreover, we know $b < c(W \cup p')$.
Using this, we get the following:
\begin{multline*}
b < c(W \cup p') = c(p') + \sum_{i \in N} \sum_{p \in P} d_i(p) \\ =
c(p') + \sum_{i \in N_{p'}} \sum_{p \in P} d_i(p) +
\sum_{i \in N \setminus N_{p'}} \sum_{p \in P} d_i(p) \\=
\sum_{i \in N_{p'}} \frac{B}{|N|}  + \sum_{i \in N \setminus N_{p'}} \sum_{p \in P} d_i(p)\\ \leq
\sum_{i \in N_{p'}} \frac{B}{|N|} + \sum_{i \in N \setminus N_{p'}} \frac{B}{|N|} = B
\end{multline*}
This concludes the proof for Sequential Phragmén.

Similarly, for the maximin support, there had to be a $p' \in \argmin_{p \in P, \text{ loads }l} s(l)$ with $c(W \cap \{p'\}) > b$. We can take the load $l$ for this argmin and again set $B = n \cdot \max_{i \in N}(l_i)$. We can again construct the prices based on the loads, which satisfies \textbf{C1}-\textbf{C4} by construction. Further, same as for the first case, a candidate $p''$ with $\sum_{i \in N_{p''}} B_i^* > c(p'')$ would be a better choice for the load $l$ and thus also for all loads. Hence, $p'$ would not have been in the argmin. Similarly, for \textbf{C6} the candidate $p_j$ would immediately lead to a better load distribution. Further, $b < B$ also holds by the same proof.
\end{proof}

\section{Local-BPJR}
\label{app:bpjr}

\citet{aziz:proportionality} defined
the following relaxation of \pjrsat (for $\sat=\satCost$).

    \begin{definition}\label{var:representation:definition:Local-BPJR-L}
        An outcome $W$ satisfies $\sat$-Local-BPJR with respect to a satisfaction function $\sat$ and an ABB instance $( A, P, c, b )$ if and only if there is no $T$-cohesive group $N'$ such that for any $W^*\subseteq P$ with  $\bigcup_{i \in N'}A_i \cap W \subsetneq W^*$ it holds that
        \[ W^* \in \argmax\left\{\sat(W') \mid W' \subseteq \bigcap_{i \in N'}A_i \land c(W') \leq c(T) \right\}\]

    \end{definition}
    
    \citet{aziz:proportionality} show that this property is satisfied by the (generalized) maximin support method (see \Cref{app:phrag}). In the unit-cost setting, $\satCost$-Local-BPJR does not imply PJR, as the following example illustrates.

    \begin{example}\label{var:representation:example:3}
        Consider the unit-cost ABB instance with three voters, projects $P=\{p_1,p_2,p_3,p_4\}$, and the following approval sets: 
        $A_1 = A_2 = \{p_1, p_2, p_3\}$ and $A_3 = \{p_1, p_2\}$. Let  $b=2$. 
        Then $W = \{p_3, p_4\}$ satisfies $\satCost$-Local-BPJR, as for any $T$-cohesive group $N'$, such that 
        \[T \in \argmax\{\sat(W') \mid W' \subseteq \bigcap_{i \in N'}A_i \land c(W') \leq \frac{\lvert N'\rvert \cdot 2}{n}\},\] it either holds that
        \begin{itemize}
            \item $\bigcup_{i \in N'} A_i \cap W \not\subseteq \bigcap_{i \in N'} A_i$ or 
            \item $|\bigcup_{i \in N'} A_i \cap W| \geq |T|$, thus maximal
        \end{itemize} 
        
        \noindent However, $W$ does not satisfy PJR because 
        \begin{itemize}
            \item $N = \{1,2,3\}$ is $\{p_1,p_2\}$-cohesive and
            \item $|\bigcup_{i \in N} A_i \cap W| = 1 < 2 = \satCost(\{p_1,p_2\})$.
        \end{itemize} 
    \end{example}

In the unit-cost setting, \pjrosat and \pjrxsat are equivalent to \pjrsat for cost-neutral and strictly increasing $\sat$. Therefore, \pjrosat and \pjrxsat are not implied by $\sat$-Local-BPJR either.

The following example shows that, for $\sat=\satCost$, \pjrosat does not imply $\sat$-Local-BPJR; as a result, the two concepts are incomparable.

\begin{example}
Consider the following ABB instance $( A, P, c, b )$
with 1 voter, three projects and budget $b = 4$:
\begin{center}
\begin{tabular}{ c|c c c}
                & $p_1$ & $p_2$ & $p_3$  \\
      \hline
      $c(\cdot)$     & 2 & 2 & 3 \\ 
      $\mu(\cdot)$     & 1 & 1 & 1 \\  
\end{tabular}
\end{center}

We observe that the single voter is $\{p_1,p_2\}$-cohesive. We claim that $\{p_1\}$
satisfies \pjro{$\satCost$}. Indeed, $p_3 \in A_v \setminus \{p_1\}$ 
and $\satCost(\{p_1,p_3\}) = 5 > 4 = \satCost(\{p_1,p_2\})$.
On the other hand, $\{p_1\}$ does not satisfy $\satCost$-Local-BPJR
as $A_v \cap \{p_1\} \subsetneq \{p_1,p_2\}$ and 
\begin{align*}
    &\{p_1,p_2\} \in \\&\argmax\left\{\satCost(W') \mid W' \subseteq \bigcap_{v \in S}A_v \land c(W') \leq c(\{p_1,p_2\}) \right\}
\end{align*}
\end{example}

\subsection{Relation between $\sat$-PJR-x and $\sat$-Local-BPJR}

We first show that \pjrxsat is a strengthening of $\sat$-Local-BPJR:
    
    \begin{proposition} \label{var:representation:proposition:PJR-w-implies-Local-BPJR-w}
        Given an ABB instance $( A, P, c, b)$ and a satisfaction function $\sat$, if $W$ satisfies \pjrxsat, then it also satisfies $\sat$-Local-BPJR.
    \end{proposition}

\begin{proof}
Assume that $W$ satisfies \pjrxsat. For the sake of a contradiction, assume that 
$\sat$-Local-BPJR is violated and let $N'$ be a $T$-cohesive group that witnesses this 
violation. That means that there is an outcome $W^*$ such that 
\begin{enumerate}
\item[\textit{(1)}] $\bigcup_{i \in N'} A_i \cap W \subsetneq W^*$,
\item[\textit{(2)}] $W^* \subseteq \bigcap_{i\in N'} A_i$, and
\item[\textit{(3)}] $c(W^*) \leq c(T)$.
\end{enumerate}
Condition (3) implies 
\[c(W^*) \leq \frac{|N'|b}{n},\] which, together with (2)
means that $N'$ is $W^*$-cohesive. Finally, from (1) it follows that there 
is a $p \in W^* \setminus W$, which by (2) must also be in $\bigcap{i\in N'} A_i$.
This means that $N'$ is a $W^*$-cohesive group such that there is a
$p \in \bigcap{i\in N'} A_i \setminus W$ for which
\[(W \cap \bigcup_{i\in N'} A_i) \cup \{p\} \subseteq W^*\]
and hence by the definition of an approval-based satisfaction function 
\[\mu((W \cap \bigcup_{i\in N'} A_i) \cup \{p\}) \leq \mu(W^*).\]
This contradicts the assumption that $W$ satisfies \pjrxsat.
\end{proof}

Moreover, \pjrx{$\satCost$} (i.e., PJR-x for the cost-based satisfaction function $\satCost$)  on its own implies $\sat$-Local-BPJR for \textit{all} 
satisfaction functions $\sat$.

    \begin{proposition} \label{var:representation:proposition:Local-PJR-w-implies-Local-BPJR-w}
        Consider an ABB instance $( A, P, c, b)$ and $W\subseteq P$. 
        If $W$ satisfies \pjrx{$\satCost$}, then $W$ satisfies $\sat$-Local-BPJR for all satisfaction functions $\sat$.
    \end{proposition}
     
    \begin{proof}
        Assume that $W$ satisfies \pjrx{$\satCost$} and $N'$ is a $T$-cohesive group, such that 
        \[\bigcup_{i\in N'}A_i \cap W \subseteq \bigcap_{i\in N'}A_i\]
        Then by \pjrx{$\satCost$}, there is no $p \in \bigcap_{i\in N'}A_i \backslash W$, such that 
        \[\satCost(\bigcup_{i\in N'}A_i \cap W \cup \{p\}) \leq \satCost(\bigcup_{i\in N'}A_i \cap T).\]
        Therefore, there is no $W'$, such that 
        \[\bigcup_{i\in N'}A_i \cap W \subset W' \subseteq \bigcap_{i\in N'}A_i\]
        for which $c(W') \leq c(T)$ holds. Thus, $W$ satisfies $\sat$-Local-BPJR.
    \end{proof}

\end{document}